\documentclass[doublespace, onecolumn,12pt,technote]{IEEEtran}

\usepackage[dvips]{epsfig}
\usepackage[dvips]{graphicx}
\usepackage{color}
\usepackage{amsfonts}
\usepackage{amssymb}
\usepackage{theorem}
\usepackage{amsmath}
\usepackage{graphics,amsmath,amsfonts,amssymb,epsfig,latexsym,amsfonts,graphicx,amssymb}
\usepackage{setspace}
\usepackage{cite}

\normalsize

\newcommand{\trace}{{\mbox{\textrm{Tr}}}}

\newcommand{\bfepsilon}{{\mbox{\boldmath $\epsilon$}}}

\newcommand{\bfmu}{{\mbox{\boldmath $\mu$}}}

\newcommand{\st}{{\rm s.t.}}
\newcommand{\by}{\mathbf{y}}
\newcommand{\bn}{\mathbf{n}}
\newcommand{\bx}{\mathbf{x}}
\newcommand{\ba}{\mathbf{a}}
\newcommand{\bb}{\mathbf{b}}
\newcommand{\bs}{\mathbf{s}}

\newcommand{\bJ}{\mathbf{J}}
\newcommand{\bE}{\mathbf{E}}

\newcommand{\bH}{\mathbf{H}}
\newcommand{\bI}{\mathbf{I}}
\newcommand{\bU}{\mathbf{U}}

\newcommand{\bQ}{\mathbf{Q}}

\newcommand{\bg}{\mathbf{g}}
\newcommand{\bW}{\mathbf{W}}
\newcommand{\bV}{\mathbf{V}}
\newcommand{\cI}{\mathcal{I}}
\newcommand{\cK}{\mathcal{K}}
\newcommand{ \firstQopt}{ \left[
\begin{array}{cc}
1 & 0\\
0 & 0\\
\end{array}
\right]
}
\newcommand{\secondQopt}{ \left[
\begin{array}{cc}
0 & 0\\
0 & 1\\
\end{array}
\right]
}
\newcommand{\thirdQopt}{ \left[
\begin{array}{cc}
0.5 & 0.5j\\
-0.5j & 0.5\\
\end{array}
\right]
}
\newcommand{\fourthQopt}{ \left[
\begin{array}{cc}
0.5 & -0.5j\\
0.5j & 0.5\\
\end{array}
\right]
}
\newcommand{\tr}{{\rm Tr}}

\newcommand{\bd}{\mathbf{d}}

\newcommand{\bc}{\mathbf{c}}

\def\BibTeX{{\rm B\kern-.05em{\sc i\kern-.025em b}\kern-.08em
    T\kern-.1667em\lower.7ex\hbox{E}\kern-.125emX}}

\title{Linear Transceiver Design for a MIMO Interfering Broadcast Channel Achieving Max-Min Fairness}
\author{Meisam Razaviyayn, Mingyi Hong, and Zhi-Quan Luo
\thanks{M. \ Razaviyayn, M.\ Hong,
and Z.-Q.\ Luo are with the Department of Electrical and Computer
Engineering University of Minnesota, Minneapolis, MN 55455,
USA}\thanks{This research is supported in part by a research gift
from Huawei Technologies Inc, and by the Army Research Office, grant
number W911NF-09-1-0279.}}
\newtheorem{lemma}{Lemma}
\newtheorem{prop}{Proposition}
\newtheorem{thm}{Theorem}

\newtheorem{coro}{Corollary}

\theorembodyfont{\rmfamily}

\begin{document}
\maketitle
\begin{abstract}
We consider the problem of linear transceiver design to achieve
max-min fairness in a downlink MIMO multicell network. This problem
can be formulated as maximizing the minimum rate among all the users
in an interfering broadcast channel (IBC). In this paper we show
that when the number of antennas is at least two at each
of the transmitters and the receivers, the min rate maximization
problem is NP-hard in the number of users. Moreover, we develop a
low-complexity algorithm for this problem by iteratively solving a sequence of convex
subproblems, and establish its global convergence to
a stationary point of the original minimum rate maximization
problem. Numerical simulations show that this algorithm
is efficient in achieving fairness among all the users.
\end{abstract}

\vspace{-0.5cm}
\section{Introduction}
\label{sec:intro}

We consider the linear transceiver design problem in a MIMO-IBC, in
which a set of Base Stations (BSs) send data to their intended
users. Both the BSs and the users are equipped with multiple
antennas, and they share the same time/frequency resource for
transmission. The objective is to maximize the minimum rate among
all the users in the network, in order to achieve network-wide
fairness.

Providing max-min fairness has long been considered as an important
design criterion for wireless networks. Hence various algorithms that
optimize the min-rate utility in different network settings have
been proposed in the literature. References \cite{zander92a,
zander92b} are early works that studied the max-min signal to
interference plus noise ratio (SINR) power control problem and a
related SINR feasibility problem in a scalar interference channel
(IC). It was shown in \cite{zander92a,
zander92b} that for randomly generated scalar ICs,
with probability one there exists a unique optimal solution to the
max-min problem. The proposed algorithm with an additional binary
search can be used to solve the max-min fairness problem
efficiently. Recently reference \cite{tan09infocom} derived a set of
algorithms based on nonlinear Perron-Frobenius theory for the same
network setting. Differently from \cite{zander92a, zander92b}, the
proposed algorithms can also deal with individual users' power
constraints.

Apart from the scalar IC case, there have been many published 
results \cite{Yang98, bengtsson01, Wiesel06, liu11MISO, Boche06,Cai11,Schubert04} 
on the min rate maximization problem in a
multiple input single output (MISO) network, in which the BSs are
equipped with multiple antennas and the users are only equipped with
a single antenna. Reference \cite{Yang98} utilized the nonnegnative
matrix theory to study the related power control problem when the
beamformers are known and fixed. When optimizing the transmit power
and the beamformers jointly, the corresponding min-rate utility
maximization problem is non-convex. Despite the lack of convexity,
the authors of \cite{bengtsson01} showed that a semidefinite
relaxation is tight for this problem, and the optimal solution can
be constructed from the solution to a reformulated semidefinite
program. Furthermore, the authors of \cite{Wiesel06} showed that
this max-min problem can be solved by a sequence of second order
cone programs (SOCP). Reference \cite{Schubert04} identified an
interesting uplink downlink duality property, in which the downlink
min-rate maximization problem can be solved by alternating between a
downlink power update and a uplink receiver update.  In a related
work \cite{Boche06}, the authors made an interesting observation
that in a single cell MISO network, the global optimum of this
problem can be obtained by solving a (simpler) weighted sum inverse
SINR problem with a set of appropriately chosen weights. However,
this observation is only true when the receiver noise is negligible.
The authors of \cite{Cai11} extended their early results
\cite{tan09infocom} to the MISO setting with a single BS and
multiple users. A fixed-point algorithm that alternates between
power update and beamformer updates was proposed, and the nonlinear
Perron-Frobenius theory was applied to prove the convergence of the
algorithm.

Unlike the MISO case, the existing work on the max-min problem for 
MIMO networks is rather limited; see \cite{Cai11} and \cite{liu11ICC}. 
Both of these studies consider a MIMO network in which a {\it single stream} is
transmitted for each user. 
In particular, the author of \cite{liu11ICC} showed that finding the
global optimal solution for this problem is intractable (NP-hard)
when the number of antennas at each transmitter/receiver is at least
\emph{three}. They then proposed an efficient algorithm that alternates
between updating the transmit and the receive beamformers to find a
local optimal solution. The key observation is that when the users'
receive beamformers are fixed, finding the set of optimal transmit
beamformers can be again reduced to a sequence of SOCP and solved
efficiently. For more discussion of the max-min and its related
resource allocation  problems in interfering wireless networks, we
refer the readers to a recent survey \cite{hong12survey}.

In this paper, we consider a MIMO interfering broadcast network whereby there
are {\it multiple} users associated with each BS, and all the users
and the BSs are equipped with multiple antennas. Such a setting is more general
than those studied in all the works cited above. Moreover, we do not
restrict the number of transmitted data streams for each user. Recent
works that deal with linear transceiver design in this type of
network include \cite{christensen08,WMMSETSP}. However, these works
aim at optimizing differentiable system utilities such as the
weighted sum rate (WSR) utility that excludes the min-rate
utility considered in this work. To the best of our knowledge, there
is no known algorithm that can effectively compute a high quality
solution for the max-min problem in the general context considered
in this work.

The main contributions of this paper are summarized as follows. First,
we show that in the considered general setting, when there are at
least {\it two} antennas at each transmitters and the receivers, the
min-rate maximization problem is NP-hard in the number of users.
This result is a generalization of that presented in
\cite{liu11ICC}, in which the NP-hardness results require more than
{\it three} antennas at the users and BSs. We further provide a
reformulation of the original max-min problem by generalizing the
framework developed in \cite{WMMSETSP}, and design an algorithm that
computes an approximate solution to the max-min problem. The
proposed algorithm has the following desirable features: {\it i)} it
is computationally efficient, as in each step a convex optimization
problem whose solution can be obtained easily in a closed form; {\it ii)} it is guaranteed to converge to a
stationary solution of the original problem.

The rest of the paper is organized as follows. In Section
\ref{sec:System_Model}, we provide the system model. In Section
\ref{sec:NP_Hardness}, we give detailed analysis of the complexity
of the considered problem. In Section 
\ref{sec:reformulation}--\ref{sec:Algorithm}, we reformulate the problem into an equivalent
form, and propose an algorithm that iteratively optimizes the
transformed problem. In Section \ref{sec:Simulations}, numerical
experiments are provided.

{\it Notations}: For a symmetric matrix $\mathbf{X}$,
$\mathbf{X}\succeq 0$ signifies that $\mathbf{X}$ is positive
semi-definite. We use $\trace(\mathbf{X})$, $|\mathbf{X}|$ and
$\mathbf{X}^H$  to denote the trace, determinant and hermitian of a
matrix, respectively. 
We use $\mathbb{R}^{N\times M}$ and $\mathbb{C}^{N\times M}$ to
denote the set of real and complex $N\times M$ matrices. We use the
notation $j$ to denote the imaginary unit, with $j^2=-1$. We use the
expression: $0\le a \perp b\ge 0$ to indicate $a\ge0, b\ge 0,
a\times b=0$.

\section{System Model and Problem Formulation}
\label{sec:System_Model} 

Consider a cellular network with $K$ cells.
Let us assume that the cell assignment of the users has been done
and each BS is interested in serving the users in its own cell. More
precisely, we assume each BS~$k$, $k=1,2,\ldots, K$, is equipped
with $M_k$ transmit antennas and serves~$I_k$ number of users in
cell~$k$. Let us use the notation~$i_k$ to denote the $i$-th user in
cell~$k$ and $N_{i_k}$ to denote the number of receive antennas of
user~$i_k$. Also define~$\cI$ and $\cI_k$ to be the set of all users
and the set of users in cell~$k$, respectively:
\begin{align}
\cI &= \left\{i_k\mid 1\leq k\leq K, 1 \leq i\leq I_k \right\}, \quad
\cI_k = \left\{i_k\mid 1 \leq i\leq I_k \right\} \nonumber.
\end{align}
Let $\cK$ to be the set of all BSs $\cK =
\left\{1,2,\ldots,K\right\}.$ Throughout, we use $i,m$ to denote the
index for the users, and use $k,\ell$ to denote the index for the
BSs.

For the standard linear channel model, the received signal of
user~$i_k$ can be written as{\small
\begin{align}
\by_{i_k} &= \underbrace{\bH_{i_kk}\bx_{i_k}}_{\textrm{desired
signal}}+\underbrace{\sum_{m \neq i,
m=1}^{I_k}\bH_{i_kk}\bx_{m_k}}_{\textrm{intracell
interference}}+\underbrace{\sum_{\ell\neq k, \ell=1}^K \sum_{m =
1}^{I_j} \bH_{i_k\ell}\bx_{m_\ell}+\bn_{i_k}}_{\textrm{intercell
interference plus noise}},  \nonumber
\end{align}}
\hspace{-0.2cm}where $\bx_{i_k}\in \mathbb{C}^{M_k\times 1}$ and
$\by_{i_k}\in \mathbb{C}^{N_{i_k}\times1}$ are respectively the
transmitted and received signal of user~$i_k$. The matrix
$\bH_{i_k j}\in\mathbb{C}^{N_{i_k} \times M_j}$ represents the
channel response from transmitter~$j$ to receiver~$i_k$, while
$\bn_{i_k} \in \mathbb{C}^{N_{i_k}\times 1}$ denotes the complex
additive white Gaussian noise with distribution
$\mathcal{CN}(0,\sigma_{i_k}^2 \bI)$ at receiver~$i_k$.

For practical considerations, we focus on optimal {\em linear}
transmit and receive strategies that can maximize a system
utility. Specifically, let BS~$k$ use a beamforming matrix
$\mathbf{V}_{i_k}$ to send the signal vector
$\mathbf{s}_{i_k}$ to receiver~$i_k$, and suppose
receiver~$i_k$ estimates the transmitted data
vector~$\mathbf{s}_{i_k}$ by using a linear beamforming matrix
$\mathbf{U}_{i_k}$, i.e.,
\begin{align}
\mathbf{x}_{i_k} = \mathbf{V}_{i_k} \; \mathbf{s}_{i_k}, \quad
\hat{\mathbf{s}}_{i_k} = \mathbf{U}_{i_k}^H \mathbf{y}_{i_k}, \quad \forall\; i_k \in \cI , \nonumber
\end{align}
where the data vector $\mathbf{s}_{i_k} \in \mathbb{C}^{d_{i_k}
\times 1}$ is normalized so that $\mathbb{E}[\mathbf{s}_{i_k}
\mathbf{s}_{i_k}^H] = \bI$, and $\hat{\mathbf{s}}_{i_k}$ is the
estimate of $\mathbf{s}_{i_k}$ at $i$-th receiver in cell~$k$.
$\mathbf{V}_{i_k} \in \mathbb{C}^{M_k \times d_{i_k}}$ and
$\mathbf{U}_{i_k} \in \mathbb{C}^{N_{i_k} \times d_{i_k}}$ are
respectively the transmit and receive beamforming matrices used for
serving the $i$-th user in cell~$k$. Let
$\mathbf{Q}_{i_k}\in\mathbb{R}^{M_k\times M_k}, \;\bQ_{i_k}\triangleq \mathbf{V}_{i_k}
\mathbf{V}^H_{i_k}$, denote the transmit covariance matrix for user
$i_k$. Let $\bV\triangleq\{\bV_{i_k}\}_{i_k\in\mathcal{I}}$ and
$\bQ\triangleq\{\bQ_{i_k}\}_{i_k\in\mathcal{I}}$.

The mean squared error (MSE) matrix for user $i_k$ can be written
as{\small
\begin{equation}\label{EQ:MSE}
\begin{split}
\bE_{i_k}&\triangleq \mathbb{E}_{\bs, \bn}\left[(\hat{\bs}_{i_k} - \bs_{i_k})(\hat{\bs}_{i_k}-\bs_{i_k})^H \right]\\
&=(\bI-\bU_{i_k}^H \bH_{i_kk} \bV_{i_k})(\bI- \bU_{i_k}^H \bH_{i_kk}\bV_{i_k})^H\\
&+ \sum_{m_\ell \neq i_k} \bU_{i_k}^H \bH_{i_k\ell}
\bV_{m_\ell}\bV_{m_\ell}^H \bH_{i_k\ell}^H
\bU_{i_k}+\sigma_{i_k}^2\bU_{i_k}^H \bU_{i_k}.
\end{split}
\end{equation}}
\vspace{-0.2cm}Treating interference as noise, the rate of the
$i$-th user in cell~$k$ is given by{\small
\begin{align}
R_{i_k} &= \log \det \bigg(\bI +  \bH_{i_kk}\bV_{i_k} \bV_{i_k}^H
\bH_{i_kk}^H  \bigg( \sigma_{i_k}^2 \bI+ \sum_{m_\ell \neq i_k}
\bH_{i_k\ell}\bV_{m_\ell} \bV_{m_\ell}^H \bH_{i_k\ell}^H
\bigg)^{-1}
\bigg) \label{EQ:rate}\\
&= \log \det \bigg(\bI +  \bH_{i_kk}\bQ_{i_k} \bH_{i_kk}^H \bigg(
\sigma_{i_k}^2 \bI+ \sum_{m_\ell \neq i_k}
\bH_{i_k\ell}\bQ_{m_\ell}
 \bH_{i_k\ell}^H \bigg)^{-1}
\bigg).\label{EQ:rateCovariance}
\end{align}}
\hspace{-0.2cm}We will occasionally use the notations $R_{i_k}(\bV)$
(resp. $R_{i_k}(\bQ)$) to make their dependencies on $\bV$ (resp.
$\bQ$) explicit.

The problem of interest is to find the transmit beamformers
$\bV = \{\bV_{i_k}\}_{i_k\in \cI}$ such that a utility of the
system is maximized, while each BS $k$'s power budget of the form
$\sum_{i=1}^{I_k} \tr (\bV_{i_k}\bV_{i_k}^H)\leq P_k$ is satisfied.
Note that $P_k$ denotes the power budget of transmitter~$k$.
In this work, our focus is on the max-min utility function, i.e., we are
interested in solving the following problem{\small
\begin{equation}\label{EQ:MaxMin_Original}\tag{P}
\begin{split}
\max_{\{\bV_{i_k}\}_{i_k \in \cI}} \quad &\min_{i_k \in \cI} \quad R_{i_k}(\bV)\\
\st \quad & \sum_{i=1}^{I_k} \tr(\bV_{i_k}\bV_{i_k}^H)\leq P_k, \quad \forall\; k \in \cK.
\end{split}
\end{equation}}
\hspace{-0.2cm}Similar to~\cite{Wiesel06}, one can
solve~\eqref{EQ:MaxMin_Original} by solving a series  of problems of
the following type for different values of $\gamma$:{\small
\begin{equation}\label{EQ:MinPower}
\begin{split}
\min_{\{\bV_{i_k}\}_{i_k \in \cI}} \quad &\sum_{k=1}^K \sum_{i=1}^{I_k} \tr(\bV_{i_k} \bV_{i_k}^H)\\
\st \quad &R_{i_k}(\bV) \geq \gamma, \quad \forall\; i_k \in \cI\\
& \sum_{i=1}^{I_k} \tr(\bV_{i_k}\bV_{i_k}^H)\leq P_k, \quad \forall\; k \in \cK.
\end{split}
\end{equation}}
\hspace{-0.2cm}The above problem is to minimize the total power
consumption in the network subject to quality of service (QoS)
constraints. In what follows, we first study the complexity status
of problem~\eqref{EQ:MaxMin_Original} and \eqref{EQ:MinPower}. Then,
we propose an efficient algorithm for designing the beamformers
based on the maximization of the worst user performance in the
system.

\section{NP-Hardness of Optimal Beamformer Design}
\label{sec:NP_Hardness} 

In this section, we
analyze the complexity status of problem~\eqref{EQ:MaxMin_Original}
and \eqref{EQ:MinPower}. In the single input single output (SISO)
case where $M_{k} = N_{i_k} = 1,\forall\; k\in \cK, \forall\; i_k \in
\cI$, it has been shown that problem~\eqref{EQ:MaxMin_Original} and
problem~\eqref{EQ:MinPower} can be solved in polynomial time,
see \cite{luo08a} and the references therein. Furthermore, it is shown that in the multiple input
single output (MISO) case where $M_{k} > N_{i_k} =1, \forall\; k\in
\cK, \forall\; i_k \in \cI$, both problems  are still polynomial time
solvable \cite{bengtsson01, Rashid98}. In this section, we consider
the MIMO case where $M_{k} \ge 2$, and $N_{i_k} \geq 2$. We show
that unlike the above mentioned special cases, both
problems~\eqref{EQ:MaxMin_Original} and \eqref{EQ:MinPower} are
NP-hard.

In fact, it is sufficient to show that for a {\it simpler} MIMO IC
network with $K$ transceiver pairs and with each node equipped with
at least two antennas, solving the max-min problem (P) and the min-power
problem~\eqref{EQ:MinPower} are both NP-hard. For convenience, we rewrite the max-min
beamformer design problem in this $K$ user MIMO IC as an
equivalent\footnote{The equivalence is in the sense that for every optimal solution $\{\bV^*\}$ of (P) with $M_k = d_k$, there exists $\lambda^* \geq 0$ so that by defining $\bQ_{k}^* = \bV_{k}^* \bV_{k}^{*H}, \forall\; k$, the point $\{\lambda^*,\bQ^*\}$ is an optimal solution of~\eqref{problemCovariance}. Conversely, if $\{\lambda^*,\bQ^*\}$ is an optimal solution of \eqref{problemCovariance} and $\bQ_{k}^* = \bV_{k}^* \bV_{k}^{*H}, \forall\; k$, then $\bV^*$ is an optimal solution of (P).} covariance maximization form
\begin{equation}
\begin{split}
\max_{(\lambda,\bQ)} \quad &\lambda\label{problemCovariance}\\
\st \quad & \lambda \leq R_k(\bQ),~\tr(\bQ_k) \leq 1, \bQ_k\succeq
0,\quad \forall~ k=1,\cdots,K.
\end{split}
\end{equation}
where $R_k(\bQ) = \log \det \left(\bI + \bH_{kk} \bQ_k  \bH_{kk}^H (\sigma^2_k \bI + \sum_{j\neq k} \bH_{kj} \bQ_j \bH_{k j}^H)^{-1}\right)$. Note that $\lambda$ is the slack variable that is introduced to
represent the objective value of the problem. The first step towards
proving the desired complexity result is to recognize certain
special structures in the optimal solutions of the problem
\eqref{problemCovariance}.

Let us consider a $3$-user MIMO IC with two antennas at each node.
Suppose $\sigma^2_k = 1$ for all $k$ and the channels are given as{\small
\begin{align}
\bH_{ii}= \left[
\begin{array}{cc}
1 & 0\\
0 & 1\\
\end{array}
\right],\quad \forall\; i =1,2,3 \quad {\rm and} \quad \bH_{im}=
\left[
\begin{array}{cc}
0 & 2\\
2 & 0\\
\end{array}
\right], \quad \forall\; i\neq m, \quad i,m
=1,2,3.\label{EQ:Channel3User}
\end{align}}
Our first result characterizes the global optimal solutions for
problem \eqref{problemCovariance} in this special network.

\begin{lemma} \label{Lemma3User}
Suppose $K=3$ and the channels are given as \eqref{EQ:Channel3User}.
Let $\mathcal{S} = \{(\lambda^*, \bQ_1^* , \bQ_2^*,\bQ_3^*)\}$ denote the set of optimal solutions of the problem
\eqref{problemCovariance}. Then $\mathcal{S}$ can be expressed
as{\small
\begin{align}
\mathcal{S}=\left\{(1, \bQ_a^*,\bQ_a^*,\bQ_a^*),(1,
\bQ_b^*,\bQ_b^*,\bQ_b^*),(1, \bQ_c^*,\bQ_c^*,\bQ_c^*),(1,
\bQ_d^*,\bQ_d^*,\bQ_d^*)\right\}, \label{EQ:temp0lemma}
\end{align}}
where {\small$\bQ_a^* = \firstQopt$, $\bQ_b^* = \secondQopt$,
$\bQ_c^* = \thirdQopt$, and $\bQ_d^* = \fourthQopt$}.
\end{lemma}

The proof of this lemma can be found in the Appendix A. Next we proceed to
consider a $5$-user interference channel with two antennas at each
node. Again suppose $\sigma^2 = 1$ and the channels are given as
{\small
\begin{align}
\bH_{ii}&= \left[
\begin{array}{cc}
1 & 0\\
0 & 1\\
\end{array}
\right],\quad \forall\; i =1,2,3 \quad {\rm and} \quad \bH_{im}=
\left[
\begin{array}{cc}
0 & 2\\
2 & 0\\
\end{array}
\right], \quad \forall\; i\neq m, \quad i,m =1,2,3; \label{EQ:Channel5User1}\\
\bH_{ii}& = \left[\begin{array}{cc}
2 & 0\\
0 & 0\\
\end{array}\right], \forall\; i=4,5, \quad \bH_{4m} =\left[\begin{array}{cc}
1 & j\\
0 & 0\\
\end{array}\right], \forall\; m=1,2,3, \quad \bH_{5m} =  \left[\begin{array}{cc}
j & 1\\
0 & 0\\
\end{array}\right], \forall\; m=1,2,3;\label{EQ:Channel5User2}\\
\bH_{im}&= 0,\; \forall\; i=1,2,3, \; \forall\; m=4,5; \quad \bH_{im} =
0,\; \forall\; i\neq m, \; i,m = 4,5.\label{EQ:Channel5User3}
\end{align}}

Our next result characterizes the global optimal solutions for the
problem \eqref{problemCovariance} for this special case.

\begin{lemma} \label{Lemma5User}
Suppose $K=5$ and the channels are given as
\eqref{EQ:Channel5User1}--\eqref{EQ:Channel5User3}. Let
$\bQ_a^*,\ \bQ_b^*$ be defined in Lemma~\ref{Lemma3User}.
Denote the
set of optimal solutions of the problem \eqref{problemCovariance} as
$\mathcal{T}$. Then  $\mathcal{T}$ can be expressed as
\begin{align}
\mathcal{T}=\left\{(1, \bQ_a^*,\bQ_a^*,\bQ_a^*,\bQ_a^*,\bQ_a^*),(1,
\bQ_b^*,\bQ_b^*,\bQ_b^*,\bQ_b^*,\bQ_b^*)\right\}
\label{EQ:temp0lemma0}.
\end{align}
\end{lemma}
\begin{proof}
First of all, it is not hard to see that by selecting each of the
values in the optimal set $\mathcal{T}$, we get
the objective value of $\lambda^* = 1$. Therefore, it suffices to
show that for any other feasible point, we get lower objective
value. To show this, we first notice that the first three users form
an interference channel which is exactly the same as the one in
Lemma~\ref{Lemma3User}. Therefore, in order to get the minimum rate
of one, we need to use one of the optimal solutions in $\mathcal{S}$
in Lemma~\ref{Lemma3User} for $(\bQ_1,\bQ_2,\bQ_3)$. Furthermore, it
is not hard to see that using either $(\bQ_1,\bQ_2,\bQ_3) =
(\bQ_c^*,\bQ_c^*,\bQ_c^*)$ or $(\bQ_1,\bQ_2,\bQ_3) =
(\bQ_d^*,\bQ_d^*,\bQ_d^*)$ would cause high interference to either user 4
or user 5 and prevent them from achieving the communication rate of one. Therefore, the only optimal solutions are the ones in the
set $\mathcal{T}$.
\end{proof}

Using Lemma~\ref{Lemma5User}, we can discretize the variables in the
max-min problem and use it to prove the NP-hardness of the problem.
In fact, for any 5 users similar to the ones in
Lemma~\ref{Lemma5User}, there are only two possible strategies that
can maximize the minimum rate of communication: either we should
transmit on the first antenna or transmit on the second antenna.
This observation will be crucial in establishing our NP-hardness result.

\begin{thm}\label{THM:NPhardMaxMin}
For a $K$-cell MIMO interference  channel where each
transmit/receive node  is equipped with at least two antennas, the
problem of designing covariance matrices to achieve max-min fairness
is NP-hard in~$K$. More specifically, solving the following problem
is NP-hard{\small
\begin{equation}\label{EQ:thmNPhard}
\begin{split}
\max_{\{\bQ_i\}_{i=1}^K} \quad &\min_{k} \quad \log \det\bigg(\bI + \bH_{kk}\bQ_k \bH_{kk}^H
\big(\sigma^2_k \bI + \sum_{j\neq k} \bH_{kj} \bQ_j \bH_{kj}^H\big)^{-1}\bigg) \\
\st\quad  &\tr(\bQ_k)\leq P_k, \bQ_k\succeq 0,~k=1,\cdots,K.
\end{split}
\end{equation}}
\end{thm}
This theorem is proved based on a polynomial time reduction from the
3-satisfiability (3-SAT) problem which is known to be NP-complete
\cite{garey79}. The 3-SAT problem is described as follows. Given $M$
disjunctive clauses $c_1,\cdots,c_M$ defined on $N$ Boolean
variables $x_1\cdots,x_N$, i.e., $c_m=y_{m1}\vee y_{m2}\vee y_{m3}$
with $y_{mi}\in\{x_1,\cdots,x_N, \bar{x}_1,\cdots,\bar{x}_N\}$, the
problem is to check whether there exists a truth assignment for the
Boolean variables such that all the clauses are satisfied
simultaneously. The details of the proof of the theorem can be found
in Appendix B.

\begin{coro}
Under the same set up as in Theorem~\ref{THM:NPhardMaxMin}, problem~\eqref{EQ:MinPower} is NP-hard.
\end{coro}

To see why the above corollary holds, we assume the contrary. Then a binary search procedure for $\lambda$ would imply a polynomial time algorithm for (P), which would contradict the NP-hardness result of Theorem~\ref{THM:NPhardMaxMin}.

\section{The Max-Min Problem and Its Equivalent Reformulation}
\label{sec:reformulation}

The complexity results established in the previous section suggests
that it is generally not possible to solve the max-min problem
\eqref{EQ:MaxMin_Original} to its global optimality in a time that grows polynomially in $K$. Guided by this insight, we reset our goal to that of designing  computationally efficient algorithms that can compute a high quality
solution for (P). To this end, we first provide an equivalent
reformulation of problem \eqref{EQ:MaxMin_Original}, which will be
used later for our algorithm design.

Introducing a slack variable $\lambda$, the problem
\eqref{EQ:MaxMin_Original} can be equivalently written as
{\small
\begin{align}
\max_{\{\mathbf{V}_{i_k}\}_{i_k\in\mathcal{I}},\lambda}&\quad\lambda\tag{\mbox{P1}}\\
{\rm s.t.}&\quad R_{i_k}(\bV)\ge \lambda,\forall~i_k\in\mathcal{I}\nonumber\\
&\quad\sum_{i\in\mathcal{I}_k}\mbox{Tr}[\mathbf{V}_{i_k}\mathbf{V}^H_{i_k}]\le
P_k,~\forall~k\in\mathcal{K}.\nonumber
\end{align}}
In order to further simplify the problem, we need to introduce the
following lemma.
\begin{lemma}\label{Lem:RateWU}
The rate of user~$i_k$ in \eqref{EQ:rate} can also be represented as
\begin{equation}\label{EQ:RateWU}
R_{i_k} = \max_{\bU_{i_k},\bW_{i_k}} \log \det \left(\bW_{i_k}\right) - \tr \left(\bW_{i_k} \bE_{i_k}\right) + d_{i_k},
\end{equation}
where~$\bE_{i_k}$ is the MSE value of user~$i_k$ given by
\eqref{EQ:MSE}.
\end{lemma}
\begin{proof}
First, by checking the first order optimality condition
of~\eqref{EQ:RateWU} with respect to~$\bU_{i_k}$, we get{\small
\begin{align}
&\bW_{i_k}^H \left(\mathbf{J}_{i_k} \bU_{i_k}^* -
\bH_{i_kk}\bV_{i_k}\right) = 0\ \ \Longrightarrow \ \ \bU_{i_k}^* =
\mathbf{J}_{i_k}^{-1} \bH_{i_kk} \bV_{i_k}, \nonumber
\end{align}}
\hspace{-0.2cm}where $\mathbf{J}_{i_k} = \sigma_{i_k}^2 \bI +
\sum_{\ell_j \in \cI} \bH_{i_kj} \bV_{\ell_j} \bV_{\ell_j}^H
\bH_{i_kj}^H$ and $\bU_{i_k}^*$ is the optimal solution
of~\eqref{EQ:RateWU}. By plugging in the optimal value~$\bU_{i_k}^*$
in~\eqref{EQ:MSE}, we obtain $ \bE_{i_k}^{opt} = \bI - \bV_{i_k}^H
\bH_{i_kk}^H \mathbf{J}_{i_k}^{-1} \bH_{i_kk}\bV_{i_k}.$ Hence
plugging $\bE_{i_k}^{opt}$ in~\eqref{EQ:RateWU} yields{\small
\begin{align} \nonumber
&\max_{\bU_{i_k},\bW_{i_k}}\quad \log \det \left(\bW_{i_k}\right) - \tr \left(\bW_{i_k} \bE_{i_k}\right) + d_{i_k}\nonumber\\
& = \max_{\bW_{i_k}} \quad\log \det \left(\bW_{i_k}\right) - \tr \left(\bW_{i_k} \bE_{i_k}^{opt}\right) + d_{i_k}. \label{EQ:RateW}
\end{align}}
\hspace{-0.2cm}The first order optimality condition
of~\eqref{EQ:RateW} with respect to~$\bW_{i_k}$ implies $\bW_{i_k}^*
= \left(\bE_{i_k}^{opt}\right)^{-1}.$

By plugging in the optimal~$\bW_{i_k}^*$ in \eqref{EQ:RateW}, we can
write{\small
\begin{align}
&\max_{\bU_{i_k},\bW_{i_k}} \log \det \left(\bW_{i_k}\right) - \tr \left(\bW_{i_k} \bE_{i_k}\right) + d_{i_k}\nonumber\\
& = -\log \det (\bE_{i_k}^{opt}) \nonumber\\
& = - \log \det \left(\bI - \bH_{i_kk}\bV_{i_k}\bV_{i_k}^H \bH_{i_kk}^H \bJ_{i_k}^{-1}\right) \nonumber\\
& = \log \det \left(\bJ_{i_k} \left(\bJ_{i_k} - \bH_{i_kk}\bV_{i_k}\bV_{i_k}^H \bH_{i_kk}^H\right)^{-1} \right) \nonumber,
\end{align}}
\hspace{-0.1cm}which is the rate of user~$i_k$ in~\eqref{EQ:rate}.
\end{proof}
\medskip

Using the observation in Lemma~\ref{Lem:RateWU}, we consider the
following reformulated problem of \eqref{EQ:MaxMin_Original}{\small
\begin{align}
\min_{\mathbf{U},\mathbf{V},
\mathbf{W}}&\quad \max_{i_k\in\mathcal{I}}\mbox{Tr}[\mathbf{W}_{i_k}\mathbf{E}_{i_k}]
-\log\det(\mathbf{W}_{i_k})-d_{i_k}\label{eqMaxMinReformulate}\tag{Q}\\
{\rm s.t.}&\quad
\sum_{i\in\mathcal{I}_k}\mbox{Tr}[\mathbf{V}_{i_k}\mathbf{V}^H_{i_k}]\le
P_k,~\forall~k\in\mathcal{K}.\nonumber
\end{align}}
\hspace{-0.1cm}Again introducing the slack variable $\lambda$, the
above problem is equivalent to{\small
\begin{align}
\max_{\mathbf{V}, \bU, \bW,\lambda}&\quad\lambda\tag{\mbox{Q1}}\\
{\rm s.t.}&\quad \mbox{Tr}[\mathbf{W}_{i_k}\mathbf{E}_{i_k}]-\log\det(\mathbf{W}_{i_k})-d_{i_k}\le -\lambda,~\forall~i_k\in\mathcal{I}\nonumber\\
&\quad
\sum_{i_k\in\mathcal{I}_k}\mbox{Tr}[\mathbf{V}_{i_k}\mathbf{V}^H_{i_k}]\le
P_k,~\forall~k\in\mathcal{K}.\nonumber
\end{align}}
\hspace{-0.1cm}In the above formulation, we have introduced the
weight matrix variables $\mathbf{W}_{i_k}\in\mathbb{C}^{d_{i_k}\times
d_{i_k}}$ and the receive beamformer variables $\mathbf{U}_{i_k}\in
\mathbb{C}^{d_{i_k}\times N_{i_k}},~i_k\in\mathcal{I}$. Note that
for fixed receive and transmit beamformers, the MSE matrix for user
$i_k$ is a function of $\bV\triangleq\{\bV_{i_k}\}_{i_k
\in\mathcal{I}}$ and $\bU_{i_k}$, and is defined in \eqref{EQ:MSE}.
In the following analysis, $\mathbf{E}_{i_k}(\bU_{i_k},\bV)$
will be occasionally used to make the dependency of the MSE matrix
on the transmit/recieve beamformers explicit. For notational
simplicity, we further define
$\mathbf{W}\triangleq\{\mathbf{W}_{i_k}\}_{i_k\in\mathcal{I}}$ and
$\mathbf{U}\triangleq\{\mathbf{U}_{i_k}\}_{i_k\in\mathcal{I}}$. The feasible set of beamformers of BS $k$ is $\mathcal{V}_k\triangleq
\{\{\mathbf{V}_{i_k}\}_{i_k\in\mathcal{I}_k}:
\sum_{i\in\mathcal{I}_k}\mbox{Tr}[\mathbf{V}_{i_k}\mathbf{V}^H_{i_k}]\le
P_k\}$; Define the feasible set for $\mathbf{V}$ as
$\mathcal{V}\triangleq\prod_{k\in\mathcal{K}}\mathcal{V}_k$.

At this point, the precise relationship of the problems (P1) and
(Q1) (or their equivalent problems (P) and (Q)) is still not clear.
In the following, we provide a series of results that reveal an
intrinsic equivalence relationship of these two problems.

To proceed, the following definitions are needed.  The minimum mean
squared error (MMSE) receiver for user $i_k$ is defined as{\small
\begin{align}
\mathbf{U}^{\rm
mmse}_{i_k}=\left(\sum_{\ell=1}^{K}\sum_{m\in\mathcal{I}_\ell}\mathbf{H}_{i_k,\ell}
\mathbf{V}_{m_\ell}\mathbf{V}^{H}_{m_\ell}\mathbf{H}^H_{i_k,\ell}+
\sigma^2_{i_k}\mathbf{I}\right)^{-1}\mathbf{H}_{i_k,k}\mathbf{V}_{i_k}&
\triangleq\Psi_{i_k}(\mathbf{V})\label{eqPsi}.
\end{align}}
\hspace{-0.1cm}When the MMSE receiver is used, the corresponding MSE
matrix $\mathbf{E}^{\rm mmse}_{i_k}$ (which is a function of $\bV$)
is given by
\begin{align}
\mathbf{E}^{\rm
mmse}_{i_k}(\bV)=\mathbf{I}-\mathbf{V}^H_{i_k}\mathbf{H}^H_{i_k,k}\mathbf{U}^{\rm
mmse}_{i_k}\succ 0.
\end{align}
Define the inverse of the MSE matrix as
\begin{align}
\Upsilon_{i_k}(\mathbf{V})\triangleq \left(\mathbf{E}^{\rm
mmse}_{i_k}(\mathbf{V})\right)^{-1}\succ 0\label{eqUpsilon}.
\end{align}
Let $\Psi(\bV)\triangleq\{\Psi_{i_k}(\bV)\}_{i_k\in\mathcal{I}}$ and
$\Upsilon(\bV)\triangleq\{\Upsilon_{i_k}(\bV)\}_{i_k\in\mathcal{I}}$.

Our next result shows that there is a connection between
the stationary solutions (or every KKT
point) of problem (P1) and the stationary solutions of
problem (Q1). Moreover, the same connection holds for the
global optimal solutions of the two problems.

\begin{prop}\label{Prop:P1andQ1}
{\it Let $(\lambda^*,\bV^*)$ be an arbitrary KKT point of problem (P1). Then
$(\lambda^*,\bV^*, \bU^*, \bW^*)=(\lambda^*,\bV^*, \Psi(\bV^*),\Upsilon(\bV^*))$ is a
KKT point of (Q1).}
{\it Conversely, suppose $(\lambda^*,\bV^*, \Psi(\bV^*), \Upsilon(\bV^*))$ is a KKT point
of problem (Q1), then $(\lambda^*,\bV^*)$ must be a KKT point of problem (P1).}
{\it Moreover, the triple
$(\bV^*,\bU^*, \bW^*)=(\bV^*, \Psi(\bV^*), \Upsilon(\bV^*))$ is
a global optimal solution of problem (Q1) if and only if $\bV^*$ is
a global optimal solution of problem (P1).}
\end{prop}


The proof of Proposition
\ref{Prop:P1andQ1} is rather involved and thus relegated to the Appendix C.

\section{The Proposed Algorithm}
\label{sec:Algorithm}

In this section, we utilize the above equivalence relationship to
design a simple algorithm for problem (P1). Our strategy is to
compute a stationary solution of problem (Q1) instead.

We propose to update the variables $\bU,\bV$, and $\bW$ alternately.
More specifically, let us use $n$ as the iteration index. Then the
proposed algorithm alternates among the following three steps
\begin{align}
\mathbf{V}^{n+1}&\in\Phi(\mathbf{U}^{n},\mathbf{W}^{n}),\quad\mathbf{U}^{n+1}=\Psi(\mathbf{V}^{n+1}),\quad
\mathbf{W}^{n+1}=\Upsilon(\mathbf{V}^{n+1})\nonumber
\end{align}
where $\Psi(\cdot)$ and $\Upsilon(\cdot)$ are respectively given in
\eqref{eqPsi} and \eqref{eqUpsilon}. The mapping $\Phi(\cdot)$ is
defined as
\begin{align}
\mathbf{V}\in\Phi(\mathbf{U},
\mathbf{W})\quad \Longleftrightarrow\quad \mathbf{V}\in
\arg\min_{\mathbf{V}\in\mathcal{V}}\max_{i_k\in\mathcal{I}}\mbox{Tr}[\mathbf{W}_{i_k}\mathbf{E}_{i_k}]-
\log\det(\mathbf{W}_{i_k})-d_{i_k}.\nonumber
\end{align}
In words, every element in the range of the map
$\Phi(\mathbf{W},\mathbf{U})$ is an optimal solution to the problem
\begin{align}
\min_{\mathbf{V}}&\max_{i_k\in\mathcal{I}}\mbox{Tr}[\mathbf{W}_{i_k}\mathbf{E}_{i_k}]-\log\det(\mathbf{W}_{i_k})-d_{i_k}\tag{Q-V}\\
{\rm s.t.}&\quad \mathbf{V}\in\mathcal{V}.\nonumber
\end{align}

In the following, we will proceed to obtain the solution to the
problem (Q-V). Introducing a slack variable $\gamma$, the problem
(Q-V) can be equivalently written as
\begin{align}
\min_{\mathbf{V},\gamma}&\quad\gamma\nonumber\\
{\rm s.t.}&\quad \mbox{Tr}[\mathbf{W}_{i_k}\mathbf{E}_{i_k}]-\log\det(\mathbf{W}_{i_k})-d_{i_k}\le \gamma,~\forall~i_k\in\mathcal{I}\nonumber\\
&\quad\mathbf{V}\in\mathcal{V}.\nonumber
\end{align}
Utilizing the definition of the MSE matrix in \eqref{EQ:MSE}, we can
see that this problem is a convex problem, as the objective is
linear, and all the constraints are convex (in fact, quadratic).
Thus, this problem can be solved in a centralized way using
conventional optimization package. The overall algorithm is
summarized in Table~\ref{FIG:Algo}.

\begin{thm}\label{Thm:Convergence}
{\it The iterates generated by the alternating algorithm given in Table~\ref{FIG:Algo} converge to the set of KKT
solutions of the problem {\rm (P1)}. In other words,}
\[
\lim_{n \rightarrow \infty} \quad  d(\bV^n,\mathcal{S})=0
\]
{\it where $\mathcal{S}$ is the set of KKT points of {\rm(P1)} and
$d(\bV,\mathcal{S}) \triangleq \inf_{\bU\in \mathcal{S}} \|\bV -
\bU\|$.}
\end{thm}

{\it Proof:}
Let us define the value of the objective function of problem (Q) as
$G(\mathbf{V}, \mathbf{U}, \mathbf{W})$.  Due to equivalence,
$G(\mathbf{V}, \mathbf{U}, \mathbf{W})$ can also represent the value
of the objective function of problem (Q1). First, we observe that the sequence $\{G(\bV^{n},
\bU^{n}, \bW^{n})\}_{n=1}^{\infty}$ monotonically decreases and
converges. Denote its limit as $\bar{G}$.  Due to the compactness of the set $\mathcal{V}$, the iterates $\{\bV^n\}_{n=1}^{\infty}$ must have a cluster point $\bar\bV$. Let $\{\bV^{n_t}\}_{t=1}^{\infty}$ be the subsequence converging to $\bar\bV$. Since the maps $\Upsilon(\cdot)$ and $\Psi(\cdot)$ are continuous, we must have
\[
\lim_{t \rightarrow \infty} \quad (\bV^{n_t},\bU^{n_t},\bW^{n_t}) =(\bar\bV,\bar\bU,\bar\bW)  \triangleq (\bar\bV,\Psi(\bar\bV),\Upsilon(\bar\bV))
\]

First we will show that in the limit we have: $
\bar{\bV}\in\Phi(\bar{\bW},\bar{\bU}).\nonumber$ Due to the
optimality of $\bV^{n_t+1}$ and monotonic decrease of the objective function, we have that
\begin{align}
G(\bV^{n_{t+1}},\bU^{n_{t+1}},\bW^{n_{t+1}})\leq G(\bV^{n_t+1}, \bU^{n_t}, \bW^{n_t})\le G(\bV, \bU^{n_t},
\bW^{n_t}), \quad\forall~\bV\in\mathcal{V}, ~\forall~n_t.\nonumber
\end{align}
Taking the limit of both sides, we have that\footnote{Note that
taking the limit inside the objective value $G(\cdot)$ is
legitimate, as the objective function of problem
\eqref{eqMaxMinReformulate} is continuous (albeit nonsmooth).}
\begin{align}
\bar{G} = G(\bar\bV,\bar\bU,\bar\bW)\le G(\bV, \bar{\bU}, \bar{\bW}),
\quad\forall~\bV\in\mathcal{V}.\nonumber
\end{align}
Consequently, we must have $\bar{\bV}\in\Phi(\bar{\bW},
\bar{\bU})$.

The next step is to establish that $(\bar{\bV}, \bar{\bU},
\bar{\bW})=(\bar{\bV}, \Psi(\bar{\bV}), \Upsilon(\bar{\bV}))$ is a
KKT solution of (Q1). Firstly, the fact that
$\bar{\bV}\in\Phi(\bar{\bW}, \bar{\bU})$ implies that $\bar{\bV}$ is
a global optimal solution of the following convex problem
\begin{align}
&\min_{\mathbf{V},\lambda}\lambda\nonumber\\
{\rm s.t.}\;~&
\mbox{Tr}[\bar{\mathbf{W}}_{i_k}\mathbf{E}_{i_k}(\bar{\bU}_{i_k},\bV)]
-\log\det(\bar{\mathbf{W}}_{i_k})-d_{i_k}\le \lambda,~\forall~i_k\in\mathcal{I}\nonumber\\
~&\mathbf{V}\in\mathcal{V}.\nonumber
\end{align}
Consequently ($\bar{\bV}$,$\bar{\lambda})$ must satisfy the
following optimality conditions (where $\bar{\bfmu},
\bar{\bfepsilon}$ are the associated Lagrangian multipliers) {\small
\begin{align}
-\sum_{i_k\in\mathcal{I}}\bar{\mu}_{i_k}\nabla_{\bV_{m_\ell}}
\left(\trace[\bar{\bW}_{i_k}\bE_{i_k}(\bar{\bU}_{i_k},\bar{\bV})]\right)-
2\bar{\epsilon}_j\bar{\bV}_{\ell_j}&=0, \forall~m_\ell\in\mathcal{I}\nonumber\\
\sum_{i_k\in\mathcal{I}}\bar{\mu}_{i_k}&=1\nonumber\\
0\le\bar{\mu}_{i_k}\perp
-\mbox{Tr}[\bar{\bW}_{i_k}\bE(\bar{\bU}_{i_k},\bar{\bV})_{i_k}]+
\log\det(\bar{\mathbf{W}}_{i_k})+d_{i_k}-\bar{\lambda}&\ge 0, \forall~i_k\in\mathcal{I}\nonumber\\
0\le\bar{\epsilon}_{k}\perp
P_k-\sum_{i_k\in\mathcal{I}_k}\trace[\bar{\bV}_{i_k}\bar{\bV}^H_{i_k}]&\ge
0, \forall~k\in\mathcal{K}\nonumber.
\end{align}}
\hspace{-0.1cm}Similarly, using the fact that
$\bar{\bU}=\Psi(\bar{\bV})$ and $\bar{\bW}=\Upsilon(\bar{\bV})$, we
have that $\bar{\bU}$ and $\bar{\bW}$ must satisfy{\small
\begin{align}
\nabla_{\bU_{i_k}}
\left(\trace[\bar{\bW}_{i_k}\bE_{i_k}(\bar{\bU}_{i_k},\bar{\bV})]\right)&=0,
\forall~i_k\in\mathcal{I}\nonumber\\
\nabla_{{\bW}_{i_k}}\left(
\trace[\bar{\bW}_{i_k}\bE_{i_k}(\bar{\bU}_{i_k},\bar{\bV})]-\log\det(\bar{\bW}_{i_k})\right)&=0,
~\forall~i_k\in\mathcal{I}\nonumber
\end{align}}
which in turn implies that the following conditions are true
\begin{align}
-\bar{\mu}_{i_k}\nabla_{\bU_{i_k}}
\left(\trace[\bar{\bW}_{i_k}\bE_{i_k}(\bar{\bU}_{i_k},\bar{\bV})]\right)&=0,
\forall~i_k\in\mathcal{I}\nonumber\\
-\bar{\mu}_{i_k}\nabla_{{\bW}_{i_k}}\left(
\trace[\bar{\bW}_{i_k}\bE_{i_k}(\bar{\bU}_{i_k},\bar{\bV})]-\log\det(\bar{\bW}_{i_k})\right)&=0,
~\forall~i_k\in\mathcal{I}\nonumber.
\end{align}
In conclusion, we have that $(\bar{\bV}, \bar{\bU},
\bar{\bW})=(\bar{\bV}, \Psi(\bar{\bV}), \Upsilon(\bar{\bV}))$ along
with the slack variable $\bar{\lambda}$ and the multipliers
$(\bar{\bfmu},\bar{\bfepsilon})$ satisfy the KKT condition for
problem (Q1) (as expressed in
\eqref{eqKKTVP2}--\eqref{eqKKTComplimentarityP22} in the Appendix C).
This result implies that $(\bar{\bV}, \bar{\bU},
\bar{\bW})=(\bar{\bV}, \Psi(\bar{\bV}), \Upsilon(\bar{\bV}))$ is a
KKT solution to problem (Q1). Applying the result in Proposition~\ref{Prop:P1andQ1}, we
conclude that $\bar\bV$ must be a KKT point of the
original problem (P1).\\
So far we have proved that any cluster point of the iterates is a
KKT point of (P1). Since the feasible set $\mathcal{V}$ is compact,
we have $ \lim_{n\rightarrow \infty} \quad d(\bV^n,\mathcal{S}) =
0$, and this completes the proof. \hfill$\blacksquare$

Several remarks regarding 
the above results and the existing results in
\cite{WMMSETSP} are in order.

{\bf Remark 1}: The original max-min problem
\eqref{EQ:MaxMin_Original} has {\it nonsmooth} objective functions.
Consequently the proof of the equivalence relationship of problem
(P) and (Q) (i.e., Proposition
\ref{Prop:P1andQ1}) is very different from
the cases presented in \cite{WMMSETSP}. In particular, in the
reformulated problem (Q1), fixing the variable $\bV$ and solving for
variables $\bU$ and $\bW$ generally admits {\it multiple solutions}.
This is because at optimality, it is possible that not all the
constraints on $\bU$ and $\bW$ variables in (Q1) are {\it active}.
For those constraints that are inactive, their corresponding
$\bU_{i_k}$ and $\bW_{i_k}$ can take multiple values.

The possibility of the existence of multiple solutions for problem
(Q1) when fixing $\bV$ has the following consequence: {\it a)} The
proof of the equivalence relationship of the stationary solutions
becomes much involved (see the proof of Proposition
\ref{Prop:P1andQ1}); {\it b)} There is no longer an
one to one relationship between the stationary solutions of problem
(P1) and (Q1). Instead, one stationary solution of problem (P1) may
correspond to a set of stationary solutions of problem (Q1).

{\bf Remark 2}: The proof of convergence of the alternating directions
algorithm becomes more involved. Different from that of
\cite{WMMSETSP}, the conventional
convergence analysis for the block coordinate descent (BCD) algorithm no
longer applies in this context. This is because the proof for the
conventional BCD algorithm requires that at least $2$ of the $3$ subproblems 
involving block variables 
must have {\it unique} solutions (see, e.g., \cite{bertsekas99}), which is clearly not the case here. 
Furthermore, the BCD algorithm requires
that the objective function is continuously differentiable {and} the
constraints are {\it separable} among the block variables. However,
in our case the constraints of problem (Q1) are {\it coupled} among
different block variables .

\section{Simulation Results}
\label{sec:Simulations} In this section, we present some numerical
experiments comparing four different approaches for the beamformer
design in the interfering broadcast channel. The first approach for
designing the beamformers is the simple ``\textit{WMMSE}" algorithm
proposed in~\cite{WMMSETSP} for maximizing the weighted sum rate of
the system. Since the sum rate utility function is not a fair
utility function among the users, we also consider the proportional
fairness (geometric mean) utility function of the users. We use the
framework in~\cite{WMMSETSP},\cite{GroupingRazaviyayn} for maximizing the geometric mean utility
function of the system and the resulting plots are denoted by the
label~``\textit{GWMMSE}".

Another way of designing the beamformers for maximizing the
performance of the worst user in the system is to approximate the
max-min utility function. One proposed approximation for the max-min
utility function could be (see  \cite{Polak03}): $ \min_{i_k}R_{i_k}
\approx \log \left(\sum_{i_k \in \cI} \exp (-R_{i_k})\right)$.
Therefore instead of solving problem~\eqref{EQ:MaxMin_Original}, we
may maximize the above approximation of the objective by solving the
following optimization problem
\begin{equation} \label{EQ:MaxMin_Approx}
\begin{split}
\max_{\bV} \quad & \sum_{i_k \in \cI} \exp (-R_{i_k})\\
\st \quad &\sum_{i=1}^{I_k} \tr(\bV_{i_k}\bV_{i_k}^H)\leq P_k, \quad \forall\; k \in \cK.
\end{split}
\end{equation}
If we restrict ourselves to the case of~$d_{i_k} = 1, \forall\; i_k
\in \cI$, then~$\bE_{i_k}$ in~\eqref{EQ:MSE} becomes a scalar and
thus we can denote it by~$e_{i_k}$. Using the
relation~\eqref{EQ:RateWU} and plugging in the optimal value for the
matrix~$\bW_{i_k}$ yields $ R_{i_k} = \log (e_{i_k}^{-1})\nonumber
$. Plugging in this relation in~\eqref{EQ:MaxMin_Approx}, we obtain
the equivalent optimization form of~\eqref{EQ:MaxMin_Approx}:
\begin{equation} \label{EQ:MaxMin_MMSE}
\begin{split}
\min_{\bV} \quad & \sum_{i_k \in \cI} e_{i_k}\\
\st \quad &\sum_{i=1}^{I_k} \tr(\bV_{i_k}\bV_{i_k}^H)\leq P_k, \quad \forall\; k,
\end{split}
\end{equation}
which is the well-known sum MSE minimization problem and we use the
algorithm in~\cite{Serbetli04} to solve~\eqref{EQ:MaxMin_MMSE}. The
corresponding plots of this method are labeled by ~``\textit{MMSE}"
in our figures.

In our simulations, the first four plots are averaged over 50
channel realizations. In each channel realization, the channel
coefficients are drawn from the zero mean unit variance i.i.d.
Gaussian distribution.

In the first numerical experiment, we consider $K=4$ BSs, each
equipped with $M=6$ antennas. There are $I = 3$ users in each cell
where each of them is equipped with $N=2$ antennas.
Figure~\ref{FIG:K=4CDF} and Figure~\ref{FIG:K=4MinRate} respectively
represent the rate cumulative rate distribution function and the minimum
rate in the system. The SNR level is set to~$20$dB in
Figure~\ref{FIG:K=4CDF}. As these figures show, our proposed method
yields substantially more fair rate allocation in the system.

In our second set of numerical experiments in Figure~\ref{FIG:K=5CDF} and Figure~\ref{FIG:K=5MinRate}, we explore the system
with $K=5$ cells where each BS serves~$I = 3$ users. The number of
transmit and receive antennas are respectively~$M = 3$ and $N=2$.


Figure~\ref{FIG:Addonef} and Figure~\ref{FIG:AddoneR} show the convergence rate of the algorithm  while a user is joining the system. In these plots, there are 5 cells and 2 users in each cell initially and at iteration 4, another user is added to one of the cells. When the extra user is added to the system, the power for the users in the same cell is reduced by a factor of $\frac{2}{3}$ and the rest of the power is used to serve the joined user initially. The precoder of the joined user is initialized randomly. Figure~\ref{FIG:Addonef} shows the objective function of (Q) during the iterations while Figure~\ref{FIG:AddoneR} demonstrates the minimum rate of the users in the system versus the iteration number.\\

Figure~\ref{FIG:ChangeChannelf} and Figure~\ref{FIG:ChangeChannelR} represent the performance and the convergence rate of the algorithm when the channel is changing during the iterations. At iteration~7, the channel is changed by a Rayleigh fade with power 0.1. As it can be seen from the plots, the algorithm converges fast and it adapts to the new channel after a few iterations.

%

\section{Appendix}\label{sec:Appendix}

\subsection{Proof of Lemma \ref{Lemma3User}}
First of all, it can be observed that choosing $\bQ_1= \bQ_2 = \bQ_3 = \bQ_a^*$ yields an objective value of $\lambda^* = 1$; the same result holds for the case of $\bQ_1= \bQ_2 = \bQ_3 = \bQ_b^*$, $\bQ_1= \bQ_2 = \bQ_3 = \bQ_c^*$, and $\bQ_1= \bQ_2 = \bQ_3 = \bQ_d^*$. 

Let $(\lambda, \bQ_1,\bQ_2,\bQ_3) \in \mathcal{S}$ be an optimal
solution. Clearly, at least one  of the users must transmit with full power, for
otherwise we could simultaneously scale $(\bQ_1,\bQ_2,\bQ_3)$ to get
a better objective function. Without loss of generality, let us
assume that user $1$ is transmitting with full power, i.e.,
$\tr(\bQ_1)=1$. Using eigenvalue decomposition of $\bQ_1$, we can
write $\bQ_1 = \alpha \ba \ba^H + \beta \bb \bb^H,$ where $\ba$ and
$\bb$ are the orthonormal eigenvectors of $\bQ_1$ and the scalars
$\alpha,\beta\geq 0$ are the eigenvalues of $\bQ_1$ with $\alpha +
\beta = 1$. Since canceling the interference results in higher rate
of communication, we have{\small
\begin{align}
R_2 &= \log \det\left(\bI + \bQ_2  \left(\bI + \sum_{m\neq 2} \bH_{2m} \bQ_m \bH_{2m}^H\right)^{-1}\right)\nonumber\\
&\leq \log \det \left(\bI + \bQ_2  \left(\bI +  \bH_{21} (\alpha \ba \ba^H + \beta \bb \bb^H) \bH_{21}^H\right)^{-1}\right)\nonumber\\
&= \log \det \left(\bI + \bQ_2  \left(\bI +  4\alpha\ \underline{\ba} \ \underline{\ba}^H  + 4 \beta\ \underline{\bb} \ \underline{\bb}^H \right)^{-1}\right)\nonumber\\
&= \log \det \left(\bI + \bQ_2  \left( \frac{1}{1+ 4\alpha}\ \underline{\ba}\ \underline{\ba}^H  + \frac{1}{1+ 4 \beta}\ \underline{\bb}\ \underline{\bb}^H \right) \right)\nonumber\\
&\leq \log \det \left(\bI + \frac{1}{\tr(\bQ_2)}\bQ_2  \left(
\frac{1}{1+ 4\alpha}\ \underline{\ba}\ \underline{\ba}^H  +
\frac{1}{1+ 4 \beta}\ \underline{\bb}\ \underline{\bb}^H \right)
\right), \label{EQ:temp1Lemma}
\end{align}}
\hspace{-0.2cm}where $\underline{\ba} = \frac{1}{2} \bH_{21} \ba$ and $\underline{\bb} = \frac{1}{2} \bH_{21} \bb$. The last inequality is due to the fact that $\tr(\bQ_2) \leq 1$. Clearly, $\underline{\ba}^H \underline{\bb} = 0$ and $\|\underline{\ba}\| = \|\underline{\bb}\| = 1$.\\

Let us use the eigenvalue decomposition $\frac{\bQ_2}{\tr(\bQ_2)} =
\theta \bc\bc^H + (1-\theta)\bd\bd^H$, for some $\theta\in [0,1]$ and some orthonormal vectors $\bc$ and $\bd$. Utilizing the fact that
determinant is the product of the eigenvalues and trace is the sum
of the eigenvalues, we can further simplify the inequality 
in~\eqref{EQ:temp1Lemma} as 
{\small
\begin{align}
R_2 &\leq \log \bigg\{1+ \tr\left[\left(\theta \bc \bc^H + (1-\theta)\bd\bd^H \right) \left(\frac{1}{1 + 4 \alpha} \underline\ba \;\underline\ba^H + \frac{1}{1 + 4 \beta} \underline\bb\; \underline\bb^H\right)\right] \nonumber\\
&\quad\quad \quad \quad + \det \left[\left(\theta \bc\bc^H + (1-\theta) \bd\bd^H\right)\left(\frac{1}{1 + 4 \alpha} \underline\ba\;\underline\ba^H + \frac{1}{1 + 4 \beta} \underline\bb\; \underline\bb^H\right)\right]\bigg\} \nonumber\\
&= \log \bigg[ 1 + \frac{\theta x}{1 + 4\alpha}+\frac{\theta (1-x)}{1 + 4\beta} + \frac{(1-\theta)(1-x)}{ 1 + 4\alpha} + \frac{(1-\theta)x}{1 + 4\beta} + \frac{\theta (1-\theta)}{(1 + 4\alpha)(1+4\beta)}\bigg]\nonumber\\
&\leq \max_{(x,\theta,\alpha,\beta) \in \mathcal{Y}} \; \log \bigg[
1 + \frac{\theta x}{1 + 4\alpha}\frac{\theta (1-x)}{1 + 4\beta} +
\frac{(1-\theta)(1-x)}{ 1 + 4\alpha} + \frac{(1-\theta)x}{1 +
4\beta} + \frac{\theta (1-\theta)}{(1 + 4\alpha)(1+4\beta)}\bigg],
\label{EQ:Temp2halflemma}
\end{align}}
\hspace{-0.1cm}where $x \triangleq |\bc^H \ba|^2$,  $\mathcal{Y}
\triangleq \{(x,\theta,\alpha,\beta)\mid \alpha+\beta = 1, \;0\leq
\alpha,\beta,x\leq 1\}$. 
Since the function in \eqref{EQ:Temp2halflemma} is linear in $x$, it
suffices to only check the boundary points $x = 0$ and $x=1$ in
order to find the maximum. The claim is that the maximum in
\eqref{EQ:Temp2halflemma} takes the value of $1$, and it is achieved
at both boundary points.

First consider the boundary point $x = 1$. We have
\begin{align}
R_2 \leq \max_{(\theta,\alpha,\beta) \in \mathcal{X}} \quad
f(\theta,\alpha,\beta),
\end{align}
where $\mathcal{X} \triangleq \{(\theta,\alpha,\beta)\mid \alpha+\beta
= 1, \;0\leq \alpha,\beta\}$ and
\begin{align}
f(\theta,\alpha,\beta) \triangleq
\log\left(1+\frac{\theta}{1+4\alpha} + \frac{1-\theta}{1+4\beta} +
\frac{\theta(1-\theta)}{(1+4\alpha)(1+4\beta)}\right)
\label{EQ:temp2Lemma}
\end{align} We are interested in finding the
set of optimal solutions of \eqref{EQ:temp2Lemma}. In particular, we
want to characterize $\mathcal{S}_1 =
\{(\theta^*,\alpha^*,\beta^*)\}$ defined by
\[
\mathcal{S}_1 \triangleq \arg \max_{(\theta,\alpha,\beta) \in
\mathcal{X}} \;\; f(\theta,\alpha,\beta).
\]
In what follows, we will prove that $ \mathcal{S}_1 =
\{(0,1,0),(1,0,1)\}.$

First we observe that $f(0,1,0) = f(1,0,1) = 1$. Now, we show that
$f(\theta,\alpha,\beta)<1, $ for all $(\theta,\alpha,\beta) \in
\mathcal{X}$ such that $0<\theta<1$. Assume the contrary that there
exists an optimal point $(\theta^*, \alpha^*, \beta^*)$ such that $ 0< \theta^* <1$.
Using the first order optimality condition $\frac{\partial
}{\partial \theta}f(\theta^*,\alpha^*,\beta^*) =0$, we obtain
\[
\theta^* = \frac{4\beta^* -4\alpha^* +1}{2}.
\]
Combining with $0 < \theta^* < 1$ yields
\begin{align}
-\frac{1}{4}< \beta^* -\alpha^* < \frac{1}{4}. \label{EQ:temp3Lemma}
\end{align}
Plugging in the value of optimal $\theta^* = \frac{4\beta^*
-4\alpha^* +1}{2}$ in $f(\cdot)$ and simplifying the equations, we
obtain
\[
f(\theta^*,\alpha^*,\beta^*) = \log \left(1+\frac{13 + 16(\beta^* -
\alpha^*)^2}{4(1+4\alpha^*)(1+4\beta^*)}\right).
\]
Combining with \eqref{EQ:temp3Lemma} yields{\small
\begin{align}
f(\theta^*,\alpha^*,\beta^*) &\leq \log\left(1 + \frac{14}{4(1+4\alpha^*)(1+4\beta^*)}\right)\nonumber\\
&\leq \log \left(1 + \frac{14}{4(1 + 4\alpha^* + 4\beta^*)}\right)\nonumber\\
& = \log\left(1 + \frac{14}{20}\right) < 1, \nonumber
\end{align}}
{\hspace{-0.2cm}}which contradicts the fact that
$\max_{(\theta,\alpha,\beta) \in \mathcal{X}} f(\theta,\alpha,\beta)
=1$. Therefore, the optimal $\theta$ only happens at the boundary
and we have
\[
\{(0,1,0),(1,0,1)\} = \arg\max_{(\theta, \alpha,\beta)\in
\mathcal{X}}\; f(\theta, \alpha,\beta).
\]

Similarly, for the case when $x = 0$, we can see that the optimal
solution set is $\{(0,0,1),(1,1,0)\}$.

Using these optimal values yields $R_2 \leq 1.$ Note that in order
to have equality $R_2 = 1$, we must have $\tr(\bQ_2) =1$ and
\[
(x,\theta,\alpha,\beta) \in
\{(1,0,1,0),(1,1,0,1),(0,0,0,1),(0,1,1,0)\}.
\]
Let us choose the optimal solution $(x,\theta,\alpha,\beta) =
(1,0,1,0)$. Therefore,
\[
\bQ_1 = \ba\ba^H,\quad \bQ_2 = \bd\bd^H, \quad x = |\bc^H\ba|^2 = 1,
\]
which yields $\underline\ba^H \bd=0$. Repeating the above argument
for user 2 and user 3, we get $\bQ_3 = \bg\bg^H$ with
$\underline\ba^H \bg= 0$. Since $\bd$ and $\bg$ are both orthogonal
to $\underline\ba$, we obtain $ \bd = \exp^{j\phi_d} \bg.$ Repeating
the above argument for the other pair of users yields
\[
\ba = \exp^{j\phi_a} \bg \quad {\rm and}\quad \underline\ba^H \ba =
0,
\]
where the last relations imply that $\ba,\bd$, and $\bg$ are the
same up to the phase rotation and they belong to the following set
(after the proper phase rotation)
\[
\ba \in \left\{ [1 \quad 0]^H,[0 \quad 1]^H,\frac{1}{\sqrt{2}}[j \quad
1]^H,\frac{1}{\sqrt{2}}[1\quad  j]^H\right\}.
\]
Each of these points gives us one of the optimal covariance matrices
in \eqref{EQ:temp0lemma}.

\subsection{Proof of Theorem \ref{THM:NPhardMaxMin}}
\begin{proof}
The proof is based on a polynomial time reduction from  the
3-satisfiability (3-SAT) problem which is known to be NP-complete.
We first consider an instance of the 3-SAT problem with $n$
variables $x_1,x_2,\ldots,x_n$ and $m$ clauses $c_1,c_2,\ldots,c_m$.
For each variable $x_i$, we consider 5 users $\mathcal{X}_{1i},
\mathcal{X}_{2i},\ldots, \mathcal{X}_{5i}$ in our interference
channel. Each user is equipped with two antennas, and the channels
between the users are specified as in
\eqref{EQ:Channel5User1}--\eqref{EQ:Channel5User3}. For each clause
$c_j$, $j=1,2,\ldots, m$, we consider one user $\mathcal{C}_j$ in
the system with two antennas. In summary, we totally have $5n+m$
users in the system. Set the noise power $\sigma^2 = 1$ and the
power budget $P_k =1$ for all users. We define the channel between
the users $\mathcal{C}_i$ and $\mathcal{C}_j$ to be zero for all
$j\neq i$. Furthermore, we assume that the channel between the transmitter and receiver of user $\mathcal{C}_i$ is given by
\[
\bH_{\mathcal{C}_i\mathcal{C}_i} = \left[
\begin{array}{cc}
1 & 0\\
0 & 0\\
\end{array}
\right]
\]
Let us also assume that {\it i)} there is no interference
among the blocks of users that correspond to different variables and
{\it ii)} there is no interference from the transmitter of user
$\mathcal{C}_j$ to the receivers of users $\mathcal{X}_{1i},\ldots,
\mathcal{X}_{5i}$ for all $i=1,2,\ldots, n$; $j = 1,2,\ldots, m$.
Consider a clause $c_j: y_{j1} + y_{j2} + y_{j3}$, where $y_{j1},
y_{j2}, y_{j3} \in
\{x_1,x_2,\ldots,x_n,\overline{x}_1,\overline{x}_2,
\ldots,\overline{x}_n\}$  with $\overline{x}_i$ denoting the negation of
$x_i$. We use the following rules to define the channels from the
transmitter of user $\mathcal{X}_{ki}$ to the receiver of user
$\mathcal{C}_j$:
\begin{itemize}
\item If the variable $x_i$ appears in $c_j$, we define the channel from the transmitter of $\mathcal{X}_{1i}$ to
the receiver of $\mathcal{C}_j$ to be {\small$
\left[\begin{array}{cc}
1 & 0\\
0 & 0\\
\end{array}\right]$}.
\item If the variable $\bar{x}_i$ appears in $c_j$, we define the channel from the
transmitter of $\mathcal{X}_{1i}$ to the receiver of $\mathcal{C}_j$
to be {\small$  \left[\begin{array}{cc}
0 & 1\\
0 & 0\\
\end{array}\right]$}.
\item If $x_i$ does not appear in $c_j$, we define the channel from the transmitter of $\mathcal{X}_{1i}$ to the receiver of $\mathcal{C}_j$ to be zero.
\item The channel from transmitters of users $\mathcal{X}_{2i},\mathcal{X}_{3i},\mathcal{X}_{4i},\mathcal{X}_{5i}$ to the receiver of user $\mathcal{C}_{j}$ is zero for all $i=1,\ldots,n$ and $j = 1,2,\ldots,m$.
\end{itemize}
\textbf{[[[As an example, draw a figure for a simple clause.]]]}
Now we claim that the 3-SAT problem is satisfiable if and
only if solving the problem \eqref{EQ:thmNPhard} for the corresponding interference channel leads to the optimum value of one. To prove this fact, let us assume that the optimum value of \eqref{EQ:thmNPhard} is one. According to the Lemma~\ref{Lemma5User}, the only way to get the rate of one for users $\mathcal{X}_{kj}$, $k=1,\ldots,5$, $j = 1,\ldots, n$, is to transmit with full power either on the first antenna or on the second antenna. Now, based on the optimal solution of \eqref{EQ:thmNPhard}, we can determine the solution of the 3-SAT problem. In particular, if user $\mathcal{X}_{1i}$ is transmitting on the first antenna, we set $x_i = 0$. Otherwise, if it transmits on the second antenna, we set $x_i = 1$. By assigning values to all the variables in this way, we claim that all clauses are satisfied. We prove by contradiction. Assume the contrary that there exists a clause $c_j$ that is not satisfied, i.e., all the corresponding variables are zero. Therefore, user $\mathcal{C}_j$ gets interference on the first receive antenna from all three users corresponding to the variables appearing in $\mathcal{C}_j$. As the result, the interference power is 3. Since the noise power is one and the received signal power is 3, the SINR level for user $\mathcal{C}_j$ is $\frac{3}{1+3}$ which contradicts the fact that the minimum rate in the system is one. \\

Now we prove the other direction. Let us assume that the 3-SAT
problem is satisfiable. We claim that the optimal value of
\eqref{EQ:thmNPhard} is one. Since in each block of 5 users the
optimum value is one, it suffices to show that the objective value
of one is achievable. Now, we design the covariance matrices based
on the solution of the 3-SAT problem. If $x_i = 0$, we transmit with
full power on the first antenna of users
$\mathcal{X}_{1i},\mathcal{X}_{2i},\ldots, \mathcal{X}_{5i}$. If
$x_i = 1$, we allocate full power for transmission on the second
antenna of users~$\mathcal{X}_{1i},\mathcal{X}_{2i},\ldots,
\mathcal{X}_{5i}$. With this allocation, each user 
$\mathcal{X}_{ki}$, $k=1,\ldots,5$, $i = 1,\ldots,n$, gets the rate
of one. For all users $\mathcal{C}_j$, $j=1,2,\ldots,m$, we transmit
with full power on the first antenna. Since 3-SAT problem is
satisfiable with the given boolean allocation of the variables, for
each clause $\mathcal{C}_j$ at one of the corresponding variables
are one. Therefore, the interference level at the receiver of user
$\mathcal{C}_j$ is at most 2. Since the received signal power at the
receiver of user $\mathcal{C}_j$ is one, the SINR level is at least
$\frac{3}{1+2}=1$ which yields the rate of communication
$R_{\mathcal{C}_j} \geq 1$. Thus, all users $\mathcal{C}_j,
j=1,\ldots, m$, have rate at least one; which completes the proof of
our claim. As the result, checking whether the objective value of one is achievable for \eqref{EQ:thmNPhard} is equivalent to solving the instance of 3-SAT
problem.  Thus, problem \eqref{EQ:thmNPhard} is NP-hard.
%
\end{proof}

\subsection{Proof of Proposition \ref{Prop:P1andQ1}}
\begin{proof}
The Lagrangian of  problem (P1) can be expressed as{\small
\begin{align}
L(\bV, \lambda;
\bfmu,\bfepsilon)=\lambda+\sum_{i_k\in\mathcal{I}}\mu_{i_k}\left(R_{i_k}(\bV)-\lambda\right)+\sum_{k\in\mathcal{K}}\epsilon_k
(P_k-\sum_{i_k\in\mathcal{I}_k}\trace[\bV_{i_k}\bV^H_{i_k}])
\end{align}}
\hspace{-0.1cm}where
$\bfmu\triangleq\{\mu_{i_k}\}_{i_k\in\mathcal{I}}$ and
$\bfepsilon\triangleq\{\epsilon_{k}\}_{k\in\mathcal{K}}$ are the set
of associated optimal Lagrangian multipliers. Suppose
$(\lambda^*,\bV^*)$ is a KKT point of (P1), and
$\{\mu^*_{i_k}\}_{i_k\in\mathcal{I}}$ and
$\{\epsilon^*_{k}\}_{k\in\mathcal{K}}$ are the set of associated
optimal Lagrangian multipliers. The KKT optimality condition for
problem (P1) can be written as{\small
\begin{align}
\nabla_{\bV_{\ell_j}}L(\bV^*, \lambda^*;
\bfmu^*,\bfepsilon^*)=\sum_{i_k\in\mathcal{I}}\mu^*_{i_k}\nabla_{\bV_{\ell_j}}R_{i_k}(\bV^*)-
2\epsilon^*_j\bV^*_{\ell_j}&=\mathbf{0}, \forall~\ell_j\in\mathcal{I}_k\label{eqKKTVP1}\\
\sum_{i_k\in\mathcal{I}}\mu^*_{i_k}&=1\label{eqKKTLambdaP1}\\
0\le\mu^*_{i_k}\perp R_{i_k}(\bV^*)-\lambda^*&\ge 0, \forall~i_k\in\mathcal{I}\label{eqKKTComplimentarityP11}\\
0\le\epsilon^*_{k}\perp
P_k-\sum_{i_k\in\mathcal{I}_k}\trace[\bV^*_{i_k}(\bV^*_{i_k})^H]&\ge
0, \forall~k\in\mathcal{K}\label{eqKKTComplimentarityP12}.
\end{align}}
\hspace{-0.1cm}Similarly, the Lagrangian of problem (Q1) can be
expressed as
\begin{align}
\widehat{L}(\bV, \bU,\bW, \lambda; \bfmu,
\bfepsilon)&=\lambda-\sum_{i_k\in\mathcal{I}}\mu_{i_k}\left(\trace[\bW_{i_k}\bE_{i_k}({\bU}_{i_k},{\bV})]
-\log\det(\bW_{i_k})-d_{i_k}+\lambda\right)\nonumber\\
&\quad +\sum_{k\in\mathcal{K}}\epsilon_k (P_k-
\sum_{i_k\in\mathcal{I}_k} \trace[\bV_{i_k}\bV^H_{i_k}])\nonumber.
\end{align}
Let $(\widehat{\bV}, \widehat{\bU},
\widehat{\bW},\widehat{\lambda})$ be a KKT solution of problem (Q1),
and let $\{\widehat{\mu}_{i_k}\}_{i_k\in\mathcal{I}_k}$ and
$\{\widehat{\epsilon}_{k}\}_{k\in\mathcal{K}}$ be the set of
associated optimal Lagrangian multipliers. The KKT optimality
condition for problem (Q1) is as follows.{\small
\begin{align}
\nabla_{\bV_{m_\ell}}\widehat{L}(\widehat{\bV},\widehat{\bU},
\widehat{\bW}, \widehat{\lambda};
\widehat{\bfmu},\widehat{\bfepsilon})=-\sum_{i_k\in\mathcal{I}}\widehat{\mu}_{i_k}\nabla_{\bV_{m_\ell}}
\left(\trace[\widehat{\bW}_{i_k}\bE_{i_k}(\widehat{\bU}_{i_k},\widehat{\bV})]\right)-
2\widehat{\epsilon}_\ell\widehat{\bV}_{m_\ell}&=0, \forall~m_\ell\in\mathcal{I}\label{eqKKTVP2}\\
\nabla_{\bU_{i_k}}\widehat{L}(\widehat{\bV},\widehat{\bU},
\widehat{\bW}, \widehat{\lambda};
\widehat{\bfmu},\widehat{\bfepsilon})
=-\widehat{\mu}_{i_k}\nabla_{\bU_{i_k}}
\left(\trace[\widehat{\bW}_{i_k}\bE_{i_k}(\widehat{\bU}_{i_k},\widehat{\bV})]\right)&=0,
\forall~i_k\in\mathcal{I}\label{eqKKTUP2}\\
\nabla_{\bW_{i_k}}\widehat{L}(\widehat{\bV},\widehat{\bU},
\widehat{\bW}, \widehat{\lambda};
\widehat{\bfmu},\widehat{\bfepsilon})=-\widehat{\mu}_{i_k}\nabla_{{\bW}_{i_k}}\left(
\trace[\widehat{\bW}_{i_k}\bE_{i_k}(\widehat{\bU}_{i_k},\widehat{\bV})]-\log\det(\widehat{\bW}_{i_k})\right)&=0,
~\forall~i_k\in\mathcal{I}\label{eqKKTWP2}\\
\sum_{i_k\in\mathcal{I}}\widehat{\mu}_{i_k}&=1\label{eqKKTLambdaP2}\\
0\le\widehat{\mu}_{i_k}\perp
-\mbox{Tr}[\widehat{\bW}_{i_k}\bE(\widehat{\bU}_{i_k},\widehat{\bV})_{i_k}]+
\log\det(\widehat{\mathbf{W}}_{i_k})+d_{i_k}-\widehat{\lambda}&\ge 0, \forall~i_k\in\mathcal{I}\label{eqKKTComplimentarityP21}\\
0\le\widehat{\epsilon}_{k}\perp
P_k-\sum_{i_k\in\mathcal{I}_k}\trace[\widehat{\bV}_{i_k}\widehat{\bV}^H_{i_k}]&\ge
0, \forall~k\in\mathcal{K}\label{eqKKTComplimentarityP22}
\end{align}}
\hspace{-0.1cm}The claim is that if $\bV^*$ satisfies the KKT system
\eqref{eqKKTVP1}--\eqref{eqKKTComplimentarityP12}, then the set of
solutions $(\widehat{\bV},\widehat{\bU},
\widehat{\bW},\widehat{\lambda})=(\bV^*,\Psi(\bV^*),
\Upsilon(\bV^*),\lambda^*)$, $(\widehat{\mu}_{i_k},
\widehat{\epsilon}_{k})=({\mu}^*_{i_k}, {\epsilon}^*_{k})$ must
satisfy the KKT system
\eqref{eqKKTVP2}--\eqref{eqKKTComplimentarityP22}. The proof for
this claim consists of three steps.

{\bf Step 1}: It is easy to observe that condition
\eqref{eqKKTLambdaP2} and \eqref{eqKKTComplimentarityP22} are
satisfied. We can then verify that by letting
$\widehat{\bU}=\Psi(\widehat{\bV})$ and
$\widehat{\bW}=\Upsilon(\widehat{\bV})$, we have that{\small
\begin{align}
&\nabla_{\bU_{i_k}}
\left(\trace[\widehat{\bW}_{i_k}\bE_{i_k}(\widehat{\bU}_{i_k},\widehat{\bV})]\right)=0,\quad
\nabla_{\bW_{i_k}}
\left(\trace[\widehat{\bW}_{i_k}\bE_{i_k}(\widehat{\bU}_{i_k},\widehat{\bV})]-\log\det(\widehat{\bW}_{i_k})\right)=0.\nonumber
\end{align}}
\hspace{-0.1cm}Consequently, conditions
\eqref{eqKKTUP2}--\eqref{eqKKTWP2} are satisfied.

{\bf Step 2}: We then show that condition \eqref{eqKKTVP2} is
satisfied.

For a set of given multipliers
$\{\mu^*_{i_k}\}_{i_k\in\mathcal{I}}$, define the following two
index sets
\begin{align}
\bar{\mathcal{A}}\triangleq\{i_k|\mu^*_{i_k}=0\};\quad
{\mathcal{A}}\triangleq\{i_k|\mu^*_{i_k}>0\}.\nonumber
\end{align}
In words, the set $\mathcal{A}$ includes the users for which the
rate constraints in (P1) are active. Notice that due to the
constraint \eqref{eqKKTLambdaP1}, set set $\mathcal{A}$ must be
nonempty, i.e., $|\mathcal{A}|>0$.

According to the above defined index sets, we partition all the
users' receive beamformers into two parts $
\bU_{\mathcal{A}}\triangleq\{\bU_{i_k}\}_{i_k\in\mathcal{A}};\quad
\bU_{\bar{\mathcal{A}}}\triangleq\{\bU_{i_k}\}_{i_k\in\bar{\mathcal{A}}}$.
Define the sets $\bW_{\mathcal{A}}$, $\bW_{\bar{\mathcal{A}}}$,
$\Psi_{\mathcal{A}}(\cdot)$, $\Psi_{\bar{\mathcal{A}}}(\cdot)$,
$\Upsilon_{\mathcal{A}}(\cdot)$,
$\Upsilon_{\bar{\mathcal{A}}}(\cdot)$ and $\bfmu_{\mathcal{A}}$ and
$\bfmu_{\bar{\mathcal{A}}}$ similarly. Define the reduced Lagrangian
function as{\small
\begin{align}
L_{{\mathcal{A}}}(\bV, \lambda;
\bfmu,\bfepsilon)&=\lambda+\sum_{i_k\in{\mathcal{A}}}\mu_{i_k}\left(R_{i_k}(\bV)-\lambda\right)+\sum_{k\in\mathcal{K}}\epsilon_k
(P_k-\sum_{i_k\in\mathcal{I}_k}\trace[\bV_{i_k}\bV^H_{i_k}])\nonumber\\
\widehat{L}_{\mathcal{A}}(\bV, \bU,\bW, \lambda; \bfmu,
\bfepsilon)&=\lambda-\sum_{i_k\in\mathcal{A}}\mu_{i_k}\left(\trace[\bW_{i_k}\bE_{i_k}({\bU}_{i_k},{\bV})]
-\log\det(\bW_{i_k})-d_{i_k}+\lambda\right)\nonumber\\
&\quad -\sum_{k\in\mathcal{K}}\epsilon_k
(\sum_{i_k\in\mathcal{I}_k}\trace[\bV_{i_k}\bV^H_{i_k}]-P_k)\nonumber
\end{align}}
\hspace{-0.1cm}A key observation is that
$\widehat{L}_{\mathcal{A}}(\bV, \bU,\bW, \lambda; \bfmu,
\bfepsilon)$ is only a function of $\bU_{\mathcal{A}}$ and
$\bW_{\mathcal{A}}$, but not of $\bU_{\bar{\mathcal{A}}}$ and
$\bW_{\bar{\mathcal{A}}}$. Consequently, we can express it as
$\widehat{L}_{\mathcal{A}}(\bV, \bU_{\mathcal{A}},\bW_{\mathcal{A}},
\lambda; \bfmu, \bfepsilon)$

{\bf Step 2.1}: We show the following key identity. If
$\bU=\Psi(\bV)$ and $\bW=\Upsilon(\bV)$, then we have
\begin{align}
\nabla_{\bV}L_{\mathcal{A}}(\bV,\lambda^*;{\bfmu}^*,{\bfepsilon}^*)=
\nabla_{\bV} \widehat{L}_{\mathcal{A}}({\bV},\bU_{\mathcal{A}},
\bW_{\mathcal{A}},\lambda^*;{\bfmu}^*,{\bfepsilon}^*)\nonumber.
\end{align}

Notice the fact that
${\bU}_{i_k}=\Psi_{i_k}({\bV}),i_k\in\mathcal{A}$ and
${\bW}_{i_k}=\Upsilon({\bV}),i_k\in\mathcal{A}$ are the {\it unique}
solutions to the following two problems, respectively
\begin{align}
&\max_{\bU_{\mathcal{A}}}
\widehat{L}_{\mathcal{A}}({\bV},{\bU}_{\mathcal{A}},
{\bW}_{\mathcal{A}},\lambda;{\bfmu}^*,{\bfepsilon}^*)\nonumber\\
&\max_{\bW_{\mathcal{A}}}
\widehat{L}_{\mathcal{A}}({\bV},{\bU}_{\mathcal{A}},
{\bW}_{\mathcal{A}},\lambda;{\bfmu}^*,{\bfepsilon}^*)\nonumber.
\end{align}
This claim can be easily checked using the first order optimality
conditions of the respective problems. We note here that the
uniqueness of the solutions comes from the fact that for all
$i_k\in\mathcal{A}$, $\mu^*_{i_k}>0$, and the fact that $\bW_{i_k}$
and $\bU_{i_k}$ are the unique solutions to the following problems,
respectively.
\begin{align}
&\arg\max_{\bU_{i_k}}
-\left(\trace[{\bW}_{i_k}\bE_{i_k}({\bU}_{i_k},{\bV})]\right)\nonumber\\
&\arg\max_{\bW_{i_k}}
-\left(\trace[{\bW}_{i_k}\bE_{i_k}({\bU}_{i_k},{\bV})]-\log\det({\bW}_{i_k})\right).\nonumber
\end{align}

Moreover, plugging ${\bU}_{\mathcal{A}}=\Psi_{\mathcal{A}}({\bV})$
and ${\bW}_{\mathcal{A}}=\Upsilon_{\mathcal{A}}({\bV})$ into
$\widehat{L}_{\mathcal{A}}({\bV},{\bU}_{\mathcal{A}},
{\bW}_{\mathcal{A}},\lambda;{\bfmu}^*,{\bfepsilon}^*)$, we
obtain{\small
\begin{align}
&\widehat{L}_{\mathcal{A}}({\bV},\Psi_{\mathcal{A}}({\bV}),
\Upsilon_{\mathcal{A}}({\bV}),\lambda;{\bfmu}^*,{\bfepsilon}^*)\nonumber\\
&=\lambda-\sum_{i_k\in\mathcal{A}}\mu^*_{i_k}\left(\trace[\bW_{i_k}\bE^{\rm
mmse}_{i_k}({\bV})] -\log\det(\bW_{i_k})-d_{i_k}+\lambda\right)
+\sum_{k\in\mathcal{K}}\epsilon_k
(P_k-\sum_{i_k\in\mathcal{I}_k}\trace[\bV_{i_k}\bV^H_{i_k}])\nonumber\\
&=\lambda-\sum_{i_k\in\mathcal{A}}\mu^*_{i_k}\left(d_{i_k}
-\log\det((\bE^{\rm mmse}_{i_k}(\bV))^{-1})-d_{i_k}+\lambda\right)
+\sum_{k\in\mathcal{K}}\epsilon_k
(P_k-\sum_{i_k\in\mathcal{I}_k}\trace[\bV_{i_k}\bV^H_{i_k}])\nonumber\\
&=\lambda+\sum_{i_k\in\mathcal{A}}\mu^*_{i_k}\left(
\log\det((\bE^{\rm mmse}_{i_k}(\bV))^{-1})-\lambda\right)
+\sum_{k\in\mathcal{K}}\epsilon_k
(P_k-\sum_{i_k\in\mathcal{I}_k}\trace[\bV_{i_k}\bV^H_{i_k}])\nonumber\\
&=\lambda+\sum_{i_k\in\mathcal{A}}\mu^*_{i_k}\left(
R_{i_k}(\bV)-\lambda\right) +\sum_{k\in\mathcal{K}}\epsilon_k
(P_k-\sum_{i_k\in\mathcal{I}_k}\trace[\bV_{i_k}(\bV_{i_k})^H])\nonumber
\end{align}}
\hspace{-0.2cm}where in the last equality we have used a well known
relationship between the MSE matrix and the achievable rate:
$R_{i_k}(\bV)=-\log\det(\bE^{\rm mmse}_{i_k}(\bV))$. Consequently,
we obtain
\begin{align}
{\bU}_{\mathcal{A}}=\Psi_{\mathcal{A}}({\bV}),~
{\bW}_{\mathcal{A}}=\Upsilon_{\mathcal{A}}({\bV})\Longrightarrow\widehat{L}_{\mathcal{A}}({\bV},\bU_{\mathcal{A}},\bW_{\mathcal{A}},\lambda;{\bfmu}^*,{\bfepsilon}^*)=L_{\mathcal{A}}(\bV,\lambda;{\bfmu}^*,{\bfepsilon}^*)\nonumber.
\end{align}
Combine the above result and the uniqueness of
${\bU}_{i_k}=\Psi_{i_k}({\bV}),i_k\in\mathcal{A}$ and
${\bW}_{i_k}=\Upsilon({\bV}),i_k\in\mathcal{A}$, we can apply
Dankin's Min-Max Theorem (See \cite[Proposition B 2.5]{bertsekas99})
to obtain the desired equality
\begin{align}
\nabla_{\bV}L_{\mathcal{A}}(\bV,\lambda^*;{\bfmu}^*,{\bfepsilon}^*)=
\nabla_{\bV} \widehat{L}_{\mathcal{A}}({\bV},\bU_{\mathcal{A}},
\bW_{\mathcal{A}},\lambda^*;{\bfmu}^*,{\bfepsilon}^*)\label{eqDankins}.
\end{align}
This concludes our proof of Step 2.1.

{\bf Step 2.2}: We show that condition \eqref{eqKKTVP1} implies
condition \eqref{eqKKTVP2}.

Notice the fact that for all $i_k\in\bar{\mathcal{A}}$,
$\mu^*_{i_k}=0$, and the fact that
$\nabla_{\bV_{i_k}}R_{i_k}(\bV^*)$ takes finite value for all
$\bV^*\in\mathcal{V}$. Then condition \eqref{eqKKTVP1} is equivalent
to the following condition
\begin{align}
\nabla_{\bV_{\ell_j}}L_{\mathcal{A}}(\bV^*,\lambda;{\bfmu}^*,{\bfepsilon}^*)=0,\forall~\ell_j\in\mathcal{I}\label{eqReducedLZero}.
\end{align}

We have the following series of equalities{\small
\begin{align}
&\nabla_{\bV_{m_\ell}} \widehat{L}({\bV}^*,\bU^*,
\bW^*,\lambda^*;{\bfmu}^*,{\bfepsilon}^*)\nonumber\\
&= -\sum_{i_k\in\mathcal{I}}{\mu}^*_{i_k}\nabla_{\bV_{m_\ell}}
\left(\trace[\bW_{i^*_k}\bE_{i_k}(\bU^*_{i_k},{\bV}^*)]\right)-
2{\epsilon}^*_\ell{\bV}^*_{m_\ell}\nonumber\\
&\stackrel{(a)}=
-\sum_{i_k\in\mathcal{A}}{\mu}^*_{i_k}\nabla_{\bV_{m_\ell}}
\left(\trace[\bW^*_{i_k}\bE_{i_k}(\bU^*_{i_k},{\bV}^*)]\right)-
2{\epsilon}^*_\ell{\bV}^*_{m_\ell}\nonumber\\
&=\nabla_{\bV_{m_\ell}} \widehat{L}_{\mathcal{A}}({\bV}^*,\bU^*,
\bW^*,\lambda^*;{\bfmu}^*,{\bfepsilon}^*)\nonumber\\
&=0\nonumber
\end{align}}
\hspace{-0.1cm}where in $(a)$ we have again used the fact that
$\mu^*_{i_k}=0$ for all $i_k\in\bar{\mathcal{A}}$,
$\mathcal{A}\bigcup\bar{\mathcal{A}}=\mathcal{I}$, and the fact that
$\nabla_{\bV_{m_\ell}}
\left(\trace[\bW_{i_k}^*\bE_{i_k}(\bU_{i_k}^*,{\bV}^*)]\right)$
takes finite value for all $\bU^*=\Upsilon_{i_k}(\bV^*)$,
$\bW^*=\Psi_{i_k}(\bV^*)$, and all $\bV^*\in\mathcal{V}$; the last
equality is due to \eqref{eqDankins} and \eqref{eqReducedLZero}.
This shows that \eqref{eqKKTVP2} is true.

{\bf Step 3}: In this step, we show that condition
\eqref{eqKKTComplimentarityP11} implies
\eqref{eqKKTComplimentarityP21}.

Let ${\bU}_{i_k}=\Psi_{i_k}({\bV})$ and
${\bW}_{i_k}=\Upsilon_{i_k}({\bV})$, we have that{\small
\begin{align}
&-\mbox{Tr}[{\bW}_{i_k}\bE_{i_k}({\bU}_{i_k},{\bV})_{i_k}]+
\log\det({\mathbf{W}}_{i_k})+d_{i_k}-{\lambda}\nonumber\\
&=-\mbox{Tr}[{\bW}_{i_k}\bE_{i_k}^{\rm mmse}({\bV})]+
\log\det({\mathbf{W}}_{i_k})+d_{i_k}-{\lambda}\nonumber\\
&=-d_{i_k}+\log\det\left(({\mathbf{E}}^{\rm
mmse}_{i_k}(\bV))^{-1}\right)+d_{i_k}-\lambda\nonumber\\
&=R_{i_k}(\bV)-\lambda\label{eqStep3}
\end{align}}
\hspace{-0.1cm}In \eqref{eqKKTComplimentarityP11} we have that
$0\le\mu^*_{i_k}\perp R_{i_k}(\bV^*)-\lambda^*\ge 0$. This condition
combined with \eqref{eqStep3} ensures
\eqref{eqKKTComplimentarityP21} is true for
$(\widehat{\bV},\widehat{\bU},
\widehat{\bW},\widehat{\lambda})=(\bV^*,\Psi(\bV^*),
\Upsilon(\bV^*),\lambda^*)$.

In conclusion, we have shown that if $(\bV^*,\lambda^*)$ and
$(\bfmu^*,\bfepsilon^*)$ satisfy the KKT system
\eqref{eqKKTVP1}--\eqref{eqKKTComplimentarityP12}, then
$(\widehat{\bV},\widehat{\bU},
\widehat{\bW},\widehat{\lambda})=(\bV^*,\Psi(\bV^*),
\Upsilon(\bV^*),\lambda^*)$, $(\widehat{\mu}_{i_k},
\widehat{\epsilon}_{k})=({\mu}^*_{i_k}, {\epsilon}^*_{k})$ satisfy
the KKT system \eqref{eqKKTVP2}--\eqref{eqKKTComplimentarityP22}.\\
%

We now establish the correspondence between the global optimal solutions of the two problems. 
The proof has two main steps.

{\bf Step 1}: We first argue that for every KKT solution $(\bV^*,
\widetilde{\bU}, \widetilde{\bW})$ of problem (Q1), there is a
corresponding solution $(\bV^*, \bU^*, \bW^*)=(\bV^*, \Psi(\bV^*),
\Upsilon(\bV^*))$ that is also a KKT solution. Furthermore, it
achieves the same objective value as $(\bV^*, \widetilde{\bU},
\widetilde{\bW})$.

Again consider the equivalent reformulation (Q1). Let $(\bfmu^*,
\bfepsilon^*)$ denote a set of optimal multiplier corresponds to
solution $(\bV^*, \widetilde{\bU}, \widetilde{\bW}, \lambda^*)$,
that is, together they satisfy the KKT system
\eqref{eqKKTVP2}--\eqref{eqKKTComplimentarityP22}. Define the index
sets $\mathcal{A}$ and $\bar{\mathcal{A}}$ as $
\bar{\mathcal{A}}\triangleq\{i_k|\mu^*_{i_k}=0\};\quad
{\mathcal{A}}\triangleq\{i_k|\mu^*_{i_k}>0\}\nonumber$. We will show
that the solution $(\bV^*,\bU^*,\bW^*)=(\bV^*, \Psi(\bV^*),
\Upsilon(\bV^*))$, along with the optimal slack variable $\lambda^*$
and the multipliers $(\bfmu^*, \bfepsilon^*)$ must also satisfy the
KKT system \eqref{eqKKTVP2}--\eqref{eqKKTComplimentarityP22}.

Firstly it is easy to see that the conditions
\eqref{eqKKTComplimentarityP22} and \eqref{eqKKTLambdaP2} are
satisfied.

We then show that the conditions \eqref{eqKKTUP2}--\eqref{eqKKTWP2}
are satisfied. Observe that for $i_k\in\mathcal{A}$, conditions
\eqref{eqKKTUP2}--\eqref{eqKKTWP2} imply that{\small
\begin{align}
\nabla_{\bU_{i_k}}
\left(\trace[\widetilde{\bW}_{i_k}\bE_{i_k}(\widetilde{\bU}_{i_k},{\bV}^*)]\right)&=0\nonumber\\
\nabla_{{\bW}_{i_k}}\left(
\trace[\widetilde{\bW}_{i_k}\bE_{i_k}(\widetilde{\bU}_{i_k},{\bV}^*)]-\log\det(\widetilde{\bW}_{i_k})\right)&=0\nonumber
\end{align}}
\hspace{-0.1cm}which in turn imply that the solutions for the above
problems are uniquely given as
\begin{align}
\widetilde{\bU}_{i_k}=\Psi_{i_k}(\bV^*),
\forall~i_k\in\mathcal{A},\quad
\widetilde{\bW}_{i_k}=\Upsilon_{i_k}(\bV^*),
\forall~i_k\in\mathcal{A}.\label{eqPartialEquivalence}
\end{align}
Thus, setting
$\widetilde{\bU}_{\mathcal{A}}=\Psi_{\mathcal{A}}(\bV^*)$ and
$\widetilde{\bW}_{\mathcal{A}}=\Upsilon_{\mathcal{A}}(\bV^*)$
ensures condition \eqref{eqKKTUP2}--\eqref{eqKKTWP2} for all
$i_k\in\mathcal{A}$. Alternatively, for $i_k\in\bar{\mathcal{A}}$,
due to the fact that $\mu^*_{i_k}=0$, the conditions
\eqref{eqKKTUP2}--\eqref{eqKKTWP2} is also satisfied.

We then show that condition \eqref{eqKKTVP2} is satisfied for
solution $(\bV^*,\bU^*, \bW^*,\lambda^*)=(\bV^*, \Psi(\bV^*),
\Upsilon(\bV^*),\lambda^*)$. This is shown by utilizing the
following series of equalities similarly as in the proof of
Proposition \ref{Prop:P1andQ1}:{\small
\begin{align}
&\nabla_{\bV_{m_\ell}} \widehat{L}({\bV}^*,\bU^*,
\bW^*,\lambda^*;{\bfmu}^*,{\bfepsilon}^*)\nonumber\\
&= -\sum_{i_k\in\mathcal{I}}{\mu}^*_{i_k}\nabla_{\bV_{m_\ell}}
\left(\trace[\bW_{i^*_k}\bE_{i_k}(\bU^*_{i_k},{\bV}^*)]\right)-
2{\epsilon}^*_\ell{\bV}^*_{m_\ell}\nonumber\\
&\stackrel{(a)}=
-\sum_{i_k\in\mathcal{A}}{\mu}^*_{i_k}\nabla_{\bV_{m_\ell}}
\left(\trace[\bW^*_{i_k}\bE_{i_k}(\bU^*_{i_k},{\bV}^*)]\right)-
2{\epsilon}^*_\ell{\bV}^*_{m_\ell}\nonumber\\
&\stackrel{(b)}=
-\sum_{i_k\in\mathcal{A}}{\mu}^*_{i_k}\nabla_{\bV_{m_\ell}}
\left(\trace[\widetilde{\bW}_{i_k}\bE_{i_k}(\widetilde{\bU}_{i_k},{\bV}^*)]\right)-
2{\epsilon}^*_\ell{\bV}^*_{m_\ell}\nonumber\\
&\stackrel{(c)}=\nabla_{\bV_{m_\ell}}
\widehat{L}_{\mathcal{A}}({\bV}^*,\widetilde{\bU},
\widetilde{\bW},\lambda^*;{\bfmu}^*,{\bfepsilon}^*)=0\nonumber
\end{align}}
\hspace{-0.2cm}where in $(a)$ (resp. in $(c)$) we have again used
the fact that $\mu^*_{i_k}=0$ for all $i_k\in\bar{\mathcal{A}}$,
$\mathcal{A}\bigcup\bar{\mathcal{A}}=\mathcal{I}$, and the fact that
$\nabla_{\bV_{m_\ell}}
\left(\trace[\bW^*_{i_k}\bE_{i_k}(\bU^*_{i_k},{\bV}^*)]\right)$
takes finite value for all feasible $\bV^*, \bU^*, \bW^*$ that
satisfies the KKT system
\eqref{eqKKTVP2}--\eqref{eqKKTComplimentarityP22}; $(b)$ is due to
\eqref{eqPartialEquivalence}; the last equality is due to the
assumption that $(\bV^*, \widetilde{\bU}, \widetilde{\bW},
\lambda^*)$ together with $(\bfmu^*, \bfepsilon^*)$ satisfy
\eqref{eqKKTVP2}.

Next, it is straightforward to see that when
$(\bV^*,\bU^*,\bW^*)=(\bV^*, \Psi(\bV^*), \Upsilon(\bV^*))$, then
for all $i_k\in\mathcal{I}$,{\small
\begin{align}
&-\mbox{Tr}[\bW^*_{i_k}\bE_{i_k}(\bU^*_{i_k},{\bV}^*)]+
\log\det(\bW^*_{i_k})+d_{i_k}-{\lambda^*}\nonumber\\
&\ge
-\mbox{Tr}[\widetilde{\bW}_{i_k}\bE_{i_k}(\widetilde{\bU}_{i_k},{\bV}^*)]+
\log\det(\widetilde{\mathbf{W}}_{i_k})+d_{i_k}-{\lambda}^*\label{eqMSEConstraintCompare}.
\end{align}}
\hspace{-0.1cm}This result implies that the feasibility part of
\eqref{eqKKTComplimentarityP21} is satisfied. In order to show that
the complementarity part of \eqref{eqKKTComplimentarityP21} is also
satisfied, it is sufficient to show that for all $i_k\in\mathcal{A}$
(i.e., for all $i_k$ such that $\mu^*_{i_k}>0$),
\eqref{eqMSEConstraintCompare} achieves strict equality. This is
guaranteed by \eqref{eqPartialEquivalence}.

So far we have shown that  $(\bV^*,\bU^*,\bW^*,\lambda^*)=(\bV^*,
\Psi(\bV^*), \Upsilon(\bV^*),\lambda^*)$ along with  $(\bfmu^*,
\bfepsilon^*)$ satisfy the KKT system
\eqref{eqKKTVP2}--\eqref{eqKKTComplimentarityP22}. The last step we
need to show is this solution achieves the same objective value as
$(\bV^*, \widetilde{\bU}, \widetilde{\bW})$.

For this purpose, observe that due to the fact that
$\sum_{i_k\in\mathcal{I}_k}\mu_{i_k}^*=1$, we must have
$|\mathcal{A}|>0$. Due to complementarity condition
\eqref{eqKKTComplimentarityP12}, at least one of the constraints
\begin{align}
-\mbox{Tr}[\Upsilon_{i_k}(\bV^*)\bE_{i_k}(\Psi_{i_k}(\bV^*),{\bV}^*)]+
\log\det(\Upsilon_{i_k}(\bV^*))+d_{i_k}-{\lambda^*}\ge 0\nonumber
\end{align}
is active. This implies that
$\min_{i_k\in\mathcal{I}}-\left(\mbox{Tr}[\Upsilon_{i_k}(\bV^*)\mathbf{E}_{i_k}(
\Psi_{i_k}(\bV^*),\bV^*)]
-\log\det(\mathbf{W}^*_{i_k})-d_{i_k}\right)=\lambda^*$. Similarly,
we must also have
$\min_{i_k\in\mathcal{I}}-\left(\mbox{Tr}[\widetilde{\mathbf{W}}_{i_k}\mathbf{E}_{i_k}(
\widetilde{\bU}_{i_k},\bV^*)]
-\log\det(\widetilde{\mathbf{W}}_{i_k})-d_{i_k}\right)=\lambda^*$.

We conclude that for every KKT solution $(\bV^*, \widetilde{\bU},
\widetilde{\bW})$ of problem (Q1), $(\bV^*, \bU^*, \bW^*)=(\bV^*,
\Psi(\bV^*), \Upsilon(\bV^*))$ is also a KKT solution, and it
achieves the same objective value as $(\bV^*, \widetilde{\bU},
\widetilde{\bW})$.

{\bf Step 2}: Now we are ready to argue the equivalence of problem
(P1) and (Q1). Let us use $f({\bV})$ and $\bar{f}(\bV, \bU, \bW)$ to
denote the objective value of problem (P1) and (Q1), respectively.

Firstly, we can check by simple substitution that for any given
$\bV\in\mathcal{V}$, $f(\bV)=\bar{f}(\bV,\Psi(\bV), \Upsilon(\bV))$.

Suppose $\bV^*,\bU^*,\bW^*$ is a global optimal solution of problem
(Q1), but $\bV^*$ is not a global optimal solution of (P1). Then
there must exist a solution $\widetilde{\bV}$ such that
$f(\widetilde{\bV})> f(\bV^*)$. From the first part of the proof we
have that $\bV^*, \Psi(\bV^*), \Upsilon(\bV^*)$ is also a KKT
solution and it achieves the same objective value as
$\bV^*,\bU^*,\bW^*$. Consequently $\bV^*, \Psi(\bV^*),
\Upsilon(\bV^*)$ is also a global optimal solution for (Q1). Using
the fact that $f(\bV)=\bar{f}(\bV,\Psi(\bV), \Upsilon(\bV))$, we
conclude that
\begin{align}
\bar{f}({\bV}^*,\Psi({\bV}^*), \Upsilon({\bV}^*))&=f({\bV}^*)<
f(\widetilde{\bV})=\bar{f}(\widetilde{\bV},\Psi(\widetilde{\bV}),
\Upsilon(\widetilde{\bV}))
\end{align}

However this contradicts the global optimality of the solution
$\bV^*, \Psi(\bV^*), \Upsilon(\bV^*)$ for problem (Q1). The reverse
direction can be argued similarly.
\end{proof}

\bibliographystyle{IEEEtranS}
\bibliography{IEEEabrv,ref}

\newpage

\begin{table}[h]
\centering \caption{Pseudo code of the proposed
algorithm}\label{FIG:Algo}
\begin{tabular}{|p{3.1in}|}
\hline
\begin{itemize}
\item [1] \; Set $n=0$. Initialize $\bV^0, \bU^0$, and $\bW^0$ randomly such that the power budget constraints are satisfied.
\item [2] \; \textbf{repeat}
\item [3] $\quad \mathbf{V}^{n+1}\in\Phi(\mathbf{U}^{n},\mathbf{W}^{n})$
\item [4] $\quad \mathbf{U}^{n+1}=\Psi(\mathbf{V}^{n+1})$
\item [5] $\quad \mathbf{W}^{n+1}=\Upsilon(\mathbf{V}^{n+1})$
\item [6] $\quad n \leftarrow n+1 $
\item [7] \; \textbf{until} some convergence criterion is met
\end{itemize}
\\
\hline
\end{tabular}
\end{table}

\begin{figure}[ht!]
 \centering
\includegraphics[width=4in]{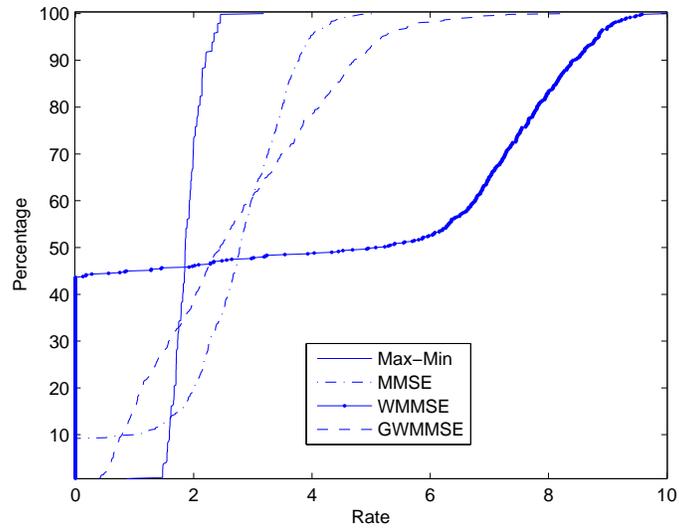}
\caption{Rate CDF: $K=4, I = 3, M=6, N=2, d=1$}\label{FIG:K=4CDF}
\end{figure}
\begin{figure}[ht!]
 \centering
\includegraphics[width=4in]{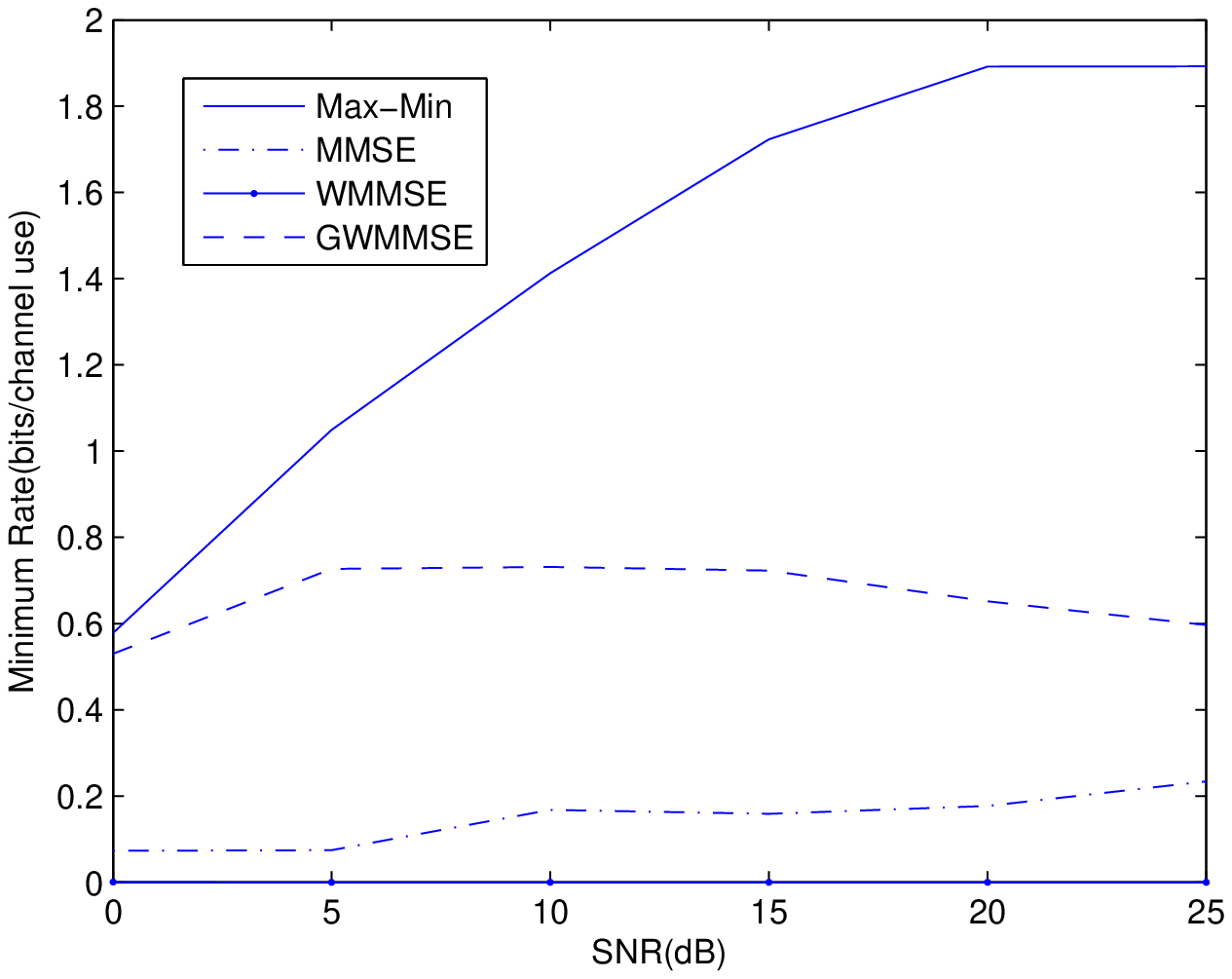}
\caption{Minimum rate in the system versus SNR: $K=4, I = 3, M=6,
N=2, d=1$}\label{FIG:K=4MinRate}
\end{figure}

\begin{figure}[ht!]
 \centering
\includegraphics[width=4in]{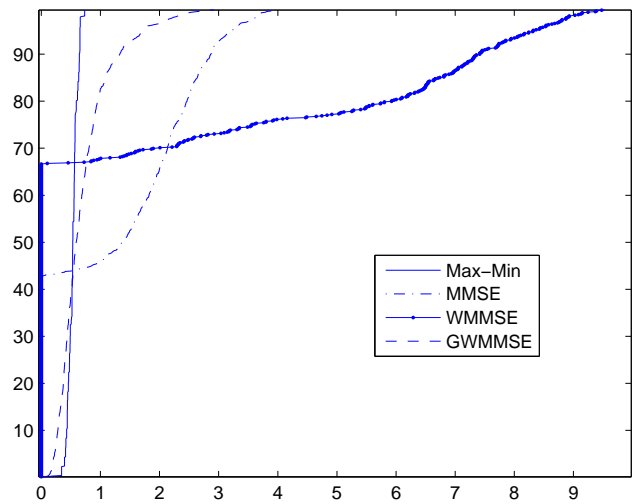}
\caption{Rate CDF: $K=5, I = 3, M=3, N=2, d=1$}\label{FIG:K=5CDF}
\end{figure}
\begin{figure}[ht!]
 \centering
\includegraphics[width=4in]{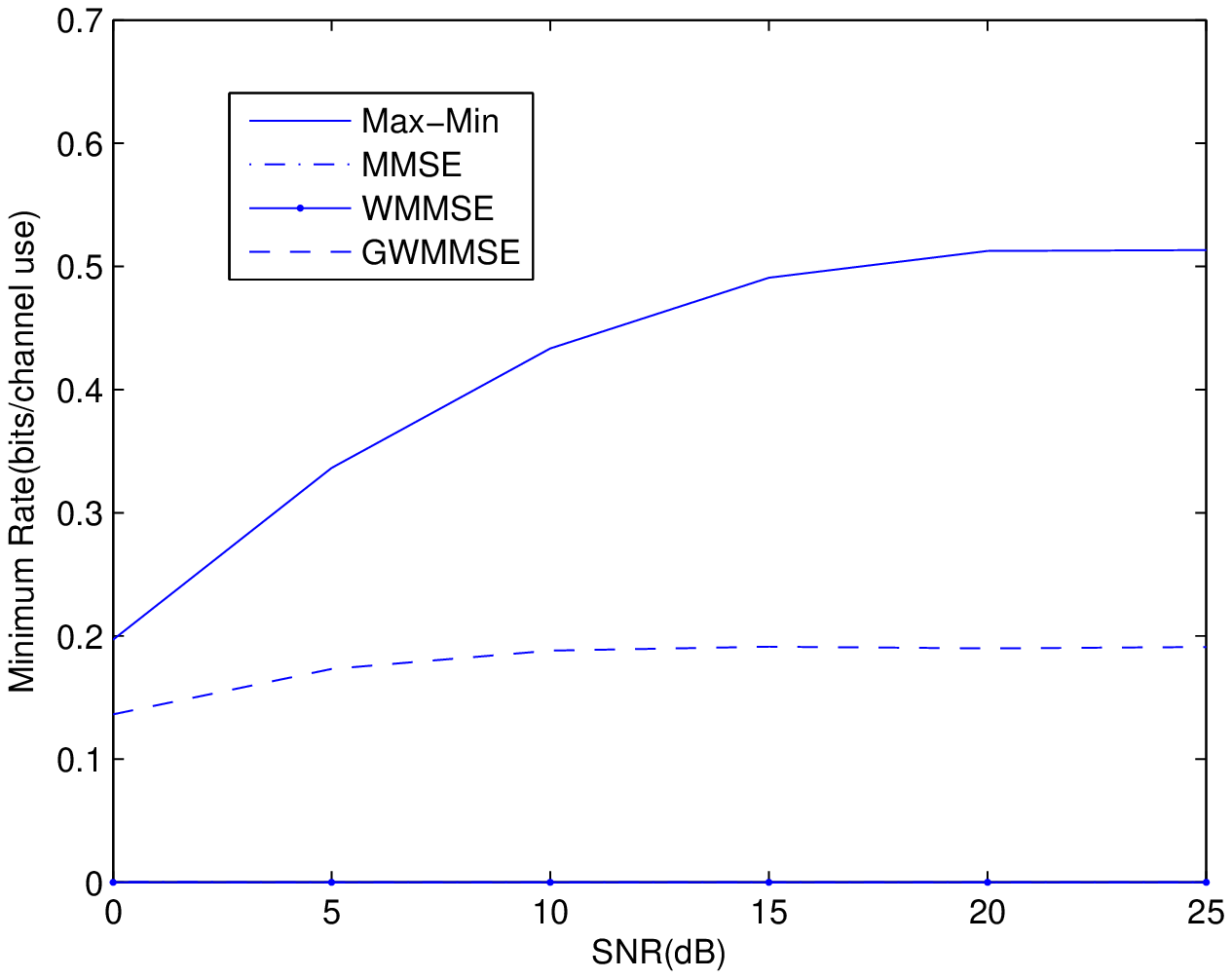}
\caption{Minimum rate in the system versus SNR: $K=5, I = 3, M=3,
N=2, d=1$}\label{FIG:K=5MinRate}
\end{figure}

\begin{figure}[ht!]
 \centering
\includegraphics[width=4in]{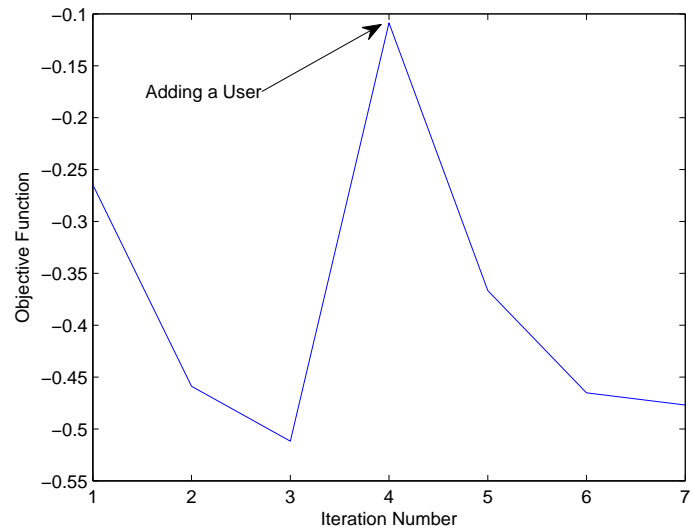}
\caption{WMMSE objective function while adding a User: $K=5, M=3, N=2, d=1$}\label{FIG:Addonef}
\end{figure}
\begin{figure}[ht!]
 \centering
\includegraphics[width=4in]{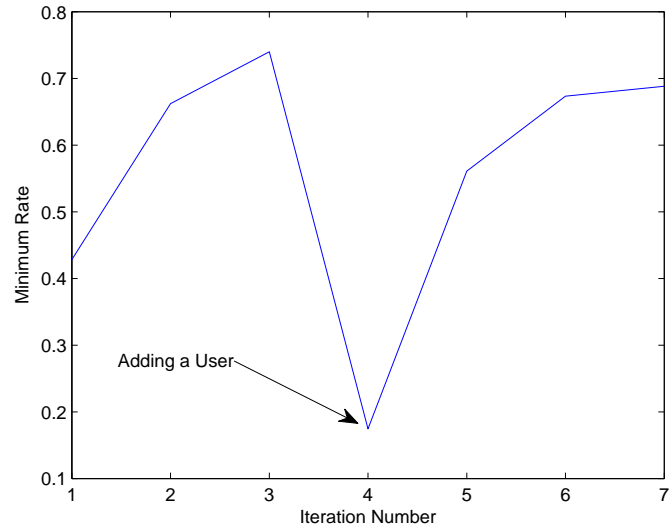}
\caption{Minimum rate while adding a User $K=5, I = 3, M=3,
N=2, d=1$}\label{FIG:AddoneR}
\end{figure}

\begin{figure}[ht!]
 \centering
\includegraphics[width=4in]{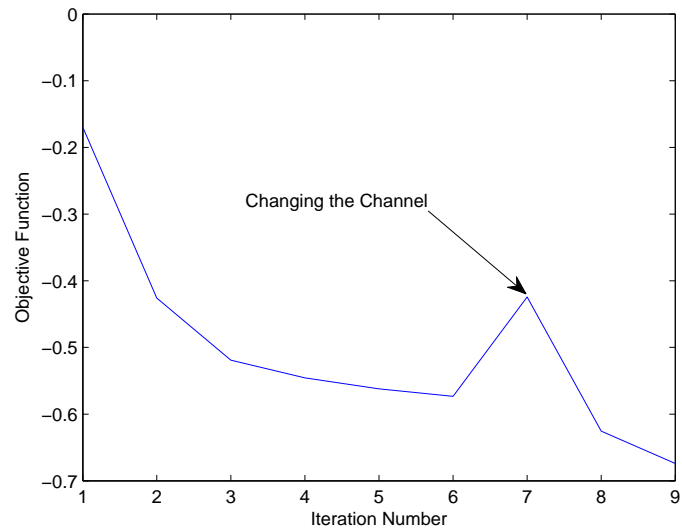}
\caption{WMMSE objective function while changing the channel: $K=5, I = 2, M=3, N=2, d=1$}\label{FIG:ChangeChannelf}
\end{figure}
\begin{figure}[ht!]
 \centering
\includegraphics[width=4in]{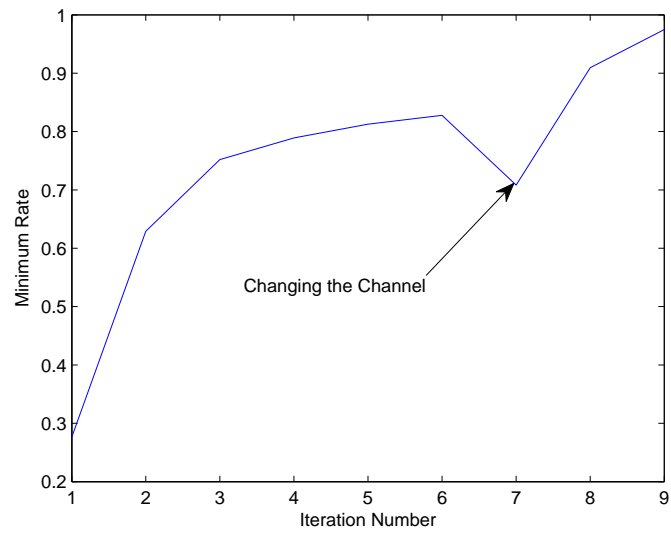}
\caption{Minimum rate while changing the channel: $K=5, I = 2, M=3,
N=2, d=1$}\label{FIG:ChangeChannelR}
\end{figure}
\end{document}